\documentclass[%
 preprint,
 amsmath,amssymb,
 aps,
pra,
]{revtex4-2}

\usepackage{graphicx}
\usepackage{dcolumn}
\usepackage{bm}
\usepackage{braket}
\usepackage{float}
\usepackage{subcaption}
\usepackage{svg}

\usepackage{graphicx}
\usepackage{tikz}
\usepackage{hyperref}
\usepackage{hypcap}
\usepackage{pgfplots}
\usepackage{subcaption}

\usepackage{amsthm}
\newtheorem{theorem}{Theorem}

\def\innerprod(#1,#2){{\left<#1\,,\,#2\right>}}
\def\Set#1{{\left\{#1\right\}}}

\def\quadtext#1{\quad\textup{#1}\quad}
\def\quadand{\quadtext{and}}
\def\pfrac#1#2{\frac{\partial #1}{\partial #2}}

\usepackage{fancyhdr}
\fancypagestyle{firststyle}
{
   \fancyhf{}
   
\fancyfoot[C]{1}
\fancyfoot[R]{\footnotesize\bf File: \jobname.tex}
}
\def\I{{\textup{I}}}
\def\II{{\textup{I\!I}}}

\def\F{{\textup{F}}}
\def\H{{\textup{H}}}

\def\ppt{{\boldsymbol p}}

\def\xia{{A}}
\def\MMan{{\cal M}}

\def\dimXi{m}
\def\uvec{{\underline u}}
\def\vvec{{\underline v}}

\def\xM{{x}}
\def\yM{{y}}

\def\tM{{t}}

\def\VE{{\boldsymbol{E}}}
\def\VB{{\boldsymbol{B}}}
\def\VD{{\boldsymbol{D}}}
\def\VH{{\boldsymbol{H}}}
\def\VP{{\boldsymbol{P}}}

\def\uvecT{{\underline{u}}}
\def\uvecF{{\underline{\underline{u}}}}
\def\svecF{{\underline{\underline{s}}}}
\def\kappabar{{\overline\kappa}}

\begin{document}


\title{Electromagnetic Boundary Conditions for Space--time Interfaces}

\date{\today}

\author{J. Gratus}
    \affiliation{Physics Department, Lancaster University, Lancaster LA1 4YB, UK\\
    The Cockcroft Institute, Sci-Tech Daresbury,
Warrington
WA4 4AD, UK}
\author{S. A. R. Horsley}
    \affiliation{School of Physics and Astronomy, Stocker Road, University of Exeter, Exeter EX4 4QL, UK}
\author{M. W. McCall}
    \affiliation{Blackett Laboratory, Imperial College, London SW7 2AZ, UK}

%
%
\begin{abstract}
    We give a general family of electromagnetic boundary conditions applicable to arbitrary space--time interfaces between electromagnetic media, which include the known space--only and time--only boundary conditions as special cases.  These boundary conditions describe a broad class of electromagnetic interfaces, including surfaces in arbitrary motion, ultra-thin (metasurface) media, and cases where the media on one or both sides of the boundary can be both spatially and temporally dispersive. Our approach utilizes 4-dimensional spacetime and addresses the question of how, and in what ways, an electromagnetic field may be connected across a 3-dimensional hypersurface.   
    We show that our proposed boundary conditions are the most general conditions consistent with causality and linearity.
\end{abstract}
\maketitle

%
%
\section{Introduction}\label{introduction}

Electromagnetic (EM) wave propagation in bulk materials is nearly always described in terms of a set of constitutive relations.  These determine the displacement field, $\boldsymbol{D}$ and magnetic flux density, $\boldsymbol{B}$, in terms of the electric $\boldsymbol{E}$ and magnetic $\boldsymbol{H}$ fields, fixing both the polarization state and wave speed within the medium~\cite{mackay2019}.  However, at the boundary of any such material, we must apply boundary conditions to connect the internal field with the outside world.  When, for example, such a bulk three--dimensional material is reduced to an ultra--thin two dimensional sheet, it becomes awkward to use the constitutive relations at all.  In this case it is simpler to characterize the material in terms of a boundary condition (BC), connecting the field vectors on either side of the sheet.  Although there has been significant work deriving these BCs for thin, static sheets of material (often practically realised as metasurfaces~\cite{achouri2015}), recent work on time--varying metamaterials and metasurfaces~\cite{galiffi2022} has opened up the possibility of more general, space--time interfaces~\cite{taravati2021}.  This leads to the natural question of what BC we should use to characterize a \emph{space--time} metasurface.  In four--dimensional language: in what ways can we connect the EM field across an arbitrary 3--dimensional hypersurface?  Here we answer this question, providing the general set of BCs, special cases of which can connect the field across any space--time interface.
%
%
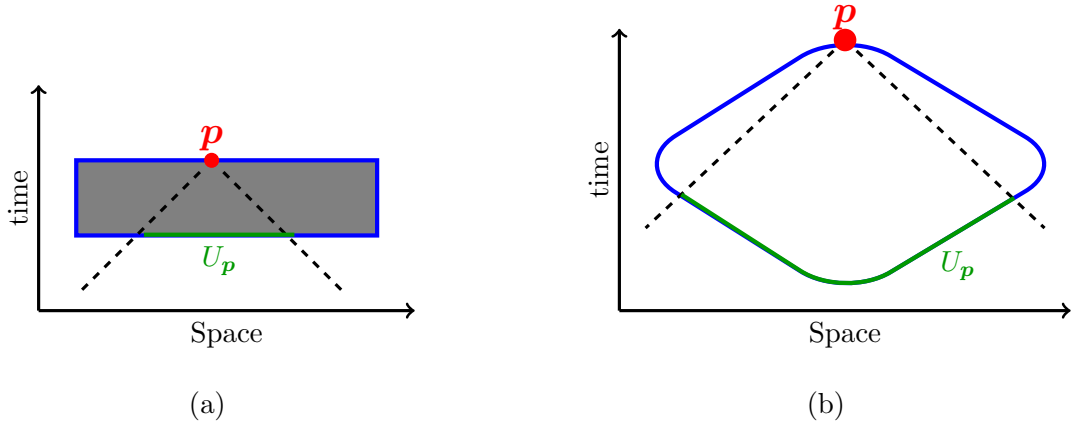
\begin{figure}
\centering
\begin{subfigure}[t]{0.5\textwidth}
\centering
\begin{tikzpicture}
\draw[very thick,->] (0,0) -- node[above,rotate=90] {time} (0,3) ;
\draw[very thick,->] (0,0) -- node[below,rotate=0] {Space} (5,0) ;
\filldraw[ultra thick,color=blue,fill=gray] 
(0.5,1) -- +(0,1) -- 
  node [pos=0.45] (A) {} +(4,1) -- +(4,0)  -- cycle ;
\fill [red] (A) circle (0.1) node[above] {\large$\ppt$}  ;
\draw [very thick,dashed] (A) -- +(-45:2.5) (A) -- +(-135:2.5) ;
\draw [ultra thick,green!60!black] (1.4,1.01) -- node[below]{$U_\ppt$} 
+(2,0.0)  ;
\end{tikzpicture}
\caption{}
\end{subfigure}%
%
\begin{subfigure}[t]{0.5\textwidth}
\centering
\begin{tikzpicture}[scale=1.5]
\draw[very thick,->] (0,0) -- node[above,rotate=90] {time} +(0,2.5) ;
\draw[very thick,->] (0,0) -- node[below,rotate=0] {Space} +(4,0) ;
\draw[ultra thick,blue,rounded corners=20 pt]
   (1.05,0.7) -- (0.1,1.3) -- node[pos=1.0] (A) {} (2,2.5) -- (4.0,1.3) -- (2,0.1) -- cycle;
\draw [very thick,dashed] (A) +(0,-.1) -- ++(-45:2.5) (A) +(0,-.1) -- +(-135:2.5);
\fill [red] (A) +(0,-.1) circle (0.1) node[above] {\large$\ppt$}  ;
  \draw [ultra thick,green!60!black,rounded corners=20pt] (0.55,1.03)  -- (2,0.1) -- (3.5,1.00);
  \draw [green!60!black] (3,0.4) node{$U_{\ppt}$} ;
 \end{tikzpicture}
 \caption{}
\end{subfigure}


\caption{\textbf{Causality and space--time boundaries:}  (a) An initially uniform homogeneous medium simultaneously changes its constitutive parameters within a bounded region (indicated as a grey rectangle), before the constitutive parameters revert back to their original form. The boundary of the grey region (blue/green), $\Sigma$, consists of the two spacelike hypersurfaces arising from the temporal change, plus the spatial boundary of the finite region.  Note that, for a point $\ppt$ on the second spacelike boundary there is only a part of the boundary, $U_\ppt$ (green) which is both in $\Sigma$ \emph{and} in the past lightcone of $\ppt$.  (b) Example showing that the EM BCs need not be applied between two bulk media.  For example, here a surface is uniformly ``switched on'' at a particular time, moves outwards symmetrically in both directions, before moving inwards and finally ``switching off''.  In both cases $U_p$ is not causally connected to $\ppt$ via paths inside $\Sigma$.  This means that information cannot propagate solely within the boundary from $U_\ppt$ to $\ppt$, implying purely local boundary conditions.  Of course information can always propagate through the bulk of the medium, such that the fields at $\ppt$ depend on all points in the past light cone, including $U_\ppt$.
%
\label{fig:schematic}}
\end{figure}

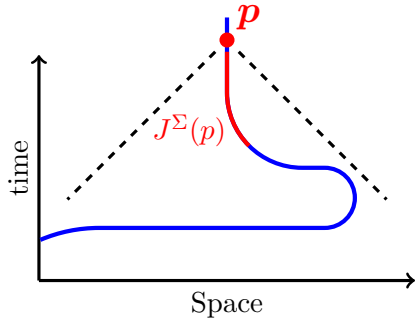
\begin{figure}
\begin{tikzpicture}
\draw[very thick,->] (-0.5,-1) -- node[above,rotate=90] {time} +(0,3) ;
\draw[very thick,->] (-0.5,-1) -- node[below,rotate=0] {Space} +(5,0) ;
\draw[ultra thick,blue] 
 (2,2.5) -- node[pos=0.3] (A){} +(0,-1) arc(-180:-90:1) -- +(0.3,0)  arc(90:-90:0.4) -- +(-3,0)  arc(90:113:2);
\fill [red] (A) circle (0.1) node[above right] {\large$\ppt$}  ;
\draw [very thick,dashed] (A) -- +(-45:3) (A) -- +(-135:3) ;
\draw[ultra thick,red] 
(A)  -- +(0,-0.7)  arc(-180:-90-45:1)  ;
\draw[red] (1.5,1) node {$J^\Sigma(p)$} ;
\end{tikzpicture}
\caption{Even though much of the boundary depicted lies in the past light cone of $p$, only the red part is in $J^\Sigma(p)$.
i.e. surface polaritons which influence $p$ can only be generated in $J^\Sigma(p)$
\label{snakecharmer} }
\end{figure}
%
%
Despite originating in earlier work on EM frequency selective surfaces~\cite{munk2005}, the metasurface concept is now `wave agnostic' in the sense that such an effective BC can be applied to modify the propagation of any linear disturbance: acoustic~\cite{assouar2018}, elastic~\cite{chen2022}, electromagnetic~\cite{chen2016}, or even thermal~\cite{wang2021}.    Here we discuss only EM metasurfaces, where a four--dimensional language is most natural.  There are several effective EM BCs that can characterise a metasurface \cite{achouri2021}: sheet transition conditions~\cite{tretyakov2003,wu2018,lebbe2023}, polarizability~\cite{holloway2009,zaluski2016},  or surface impedance~\cite{senior1995,tretyakov2003,maci2011}, with the latter being similar to Robin's BC~\cite{gustafson1998}.  For a static surface, Ref. \cite{lindell2017} considered the most general linear, local BCs that can describe any stationary metasurface. These BCs need not conserve energy, potentially describing PT-symmetric or generally non--Hermitian metasurfaces with designed distributions of loss and gain~\cite{fleury2014,tapar2021,fan2022}.  However, more recently space--time varying surfaces have been discussed and experimentally investigated~\cite{zhang2018,li2020,wang2020}, including quantum effects~\cite{oue2023} and the demonstration of synthetic motion~\cite{harwood2024}.  Here an elementary concept is the space--time interface~\cite{caloz2019a,caloz2019b,caloz2023} (also implicit in Milton's `field patterns'~\cite{milton2017}), and it is not obvious how to adapt the aforementioned BCs to this situation, existing discussions often performing calculations, case--by--case.  For instance, Mostapha et al. \cite{mostafa2024temporal} have given a comprehensive review of the BCs applied to uniform anisotropic media that are abruptly modulated, and Li et al \cite{li2022nonreciprocity} considered nonreciprocity for temporal boundaries in an anisotropic medium, analysing the temporal analog of a Faraday rotator in abruptly switched parameters in a magnetoplasma.  Meanwhile Caloz et al. \cite{caloz2022generalized} considered generalized space-time engineered isotropic metamaterials. Their classification noted that while conservation of  $\bf E$ and $\bf H$ (respectively $\bf D$ and $\bf B$) along purely spatial (temporal) interfaces, the BCs for uniformly moving interfaces modify both spatial and temporal frequencies. They also noted that more general (accelerated) interface motions may be better described within the framework of differential geometry.  Here we use the setting of four--dimensional space--time to find the general set of BCs appropriate for space--time interfaces.

Fig. \ref{fig:schematic} illustrates a typical case, where a set of BCs (the boundary in each case is indicated in blue) must be applied to connect the EM field between different space--time regions where the constitutive parameters may or may not be different.  Note that, unlike a conventional time--independent boundary, we have the interesting possibility that the past light cone of a given point, $\boldsymbol{p}$ on the boundary may be causally connected to only a \emph{portion} of the past boundary that depends on its motion.  Points on a time boundary for instance cannot be causally connected to one another at all.  This is unlike a space--boundary, where every point on the boundary is causally connected to its previous state. 
We later consider dispersive boundaries where polaritons propagate inside the boundary. Whereas such polaritons are necessarily subluminal, the motion of the spacetime boundaries themselves may be subluminal, luminal or superluminal. The nature of causality when considering such dispersive boundary conditions therefore needs to be assessed.    We note at this stage that when considering such dispersive boundaries, the points influencing the point  $\boldsymbol{p}$ can only be over regions of the spacetime boundary
that are causally connected to p via luminal paths that lie in the boundary - see Fig. \ref{snakecharmer}.

To understand how such a time variation affects the required BCs, take a concrete example: a $y$--polarized electric field $\boldsymbol{E}=E \boldsymbol{e}_{y}$ and $z$--polarized magnetic field $\boldsymbol{H}=H\boldsymbol{e}_{z}$, propagate along the $x$--axis, with a thin polarizable sheet at the time--varying position, $x(t)$.  In this system, we can write the polarization density as $P=\alpha_{\rm E}\delta(x-x(t))E$, where $\alpha_{\rm E}$ is the polarizability.  An idealised experiment could realise this through e.g. inducing a change in the permittivity in a non-linear bulk medium, illuminating it with a sheet of intense laser light and moving this sheet as a function of time.  For these assumed fields and polarization, Maxwell's equations reduce to
\begin{align}
    \frac{\partial E}{\partial x}=&-\mu_0\frac{\partial H}{\partial t}\nonumber\\
    \frac{\partial H}{\partial x}=&-\epsilon_0\frac{\partial E}{\partial t}-\frac{\partial P}{\partial t}=-\epsilon_0\frac{\partial E}{\partial t}-\alpha_{\rm E}\left[\delta(x-x(t))\frac{\partial E}{\partial t}-\dot{x}(t) E\frac{\partial }{\partial x}\delta(x-x(t))\right].
\end{align}
Integrating the second of these equations across the width of the polarizable sheet, $x\in[x(t)-\eta,x(t)+\eta]$, with $\eta\to0$ and taking the derivatives of the electric field on the sheet to be the average value between the two regions, e.g. $\left.\partial_{x}E\right|_{x(t)}=\frac{1}{2}\left(\partial_{x}E^{\I}\vert_{x(t)}+\partial_x E^{\II}\vert_{x(t)}\right)$ we find the following BCs, 
\begin{align}
    \left.E^{\I}\right|_{x(t)}=&\left.E^{\II}\right|_{x(t)}\nonumber\\[10pt]
    \left.\left[H^{\I}+\frac{\alpha_{\rm E}}{2}\left(\frac{\partial E^{\I}}{\partial t}+\dot{x}(t)\frac{\partial E^{\I}}{\partial x}\right)\right]\right|_{x(t)}=&\left.\left[H^{\II}+\frac{\alpha_{\rm E}}{2}\left(\frac{\partial E^{\II}}{\partial t}+\dot{x}(t)\frac{\partial E^{\II}}{\partial x}\right)\right]\right|_{x(t)}.\label{eq:sheet-bc}
\end{align}
Eq. (\ref{eq:sheet-bc}) is equivalent to the statement that the magnetic field is discontinuous, due to the tangential current flowing within the sheet.  The motion of the surface along the trajectory $x(t)$ changes the value of this current, and we gain an additional term proportional to the velocity of the sheet, $\dot{x}$.  For low velocities $|\dot{x}|\ll c$, the magnetic field discontinuity is proportional to the time derivative of the electric field, whereas a superluminal sheet, $|\dot{x}|\gg c$ has a discontinuity proportional to the spatial derivative of the electric field.  What is the general setting for this particular space--time boundary condition? In this paper we will answer this question. 

In section \ref{sec:bcs4d} we review the electromagnetic boundary conditions at purely spatial and temporal boundaries and extend these to arbitrarily moving boundaries and to include additional boundary conditions when the media either side of the boundary are spatially or temporally dispersive. In preparation for discussing general boundary conditions, we express these special cases in tensor notation. In section \ref{Sec:GenBdd Formula} we present our main result, i.e.  linear boundary conditions at a general spacetime boundary that embrace the preceding purely spatial and temporal boundaries as special cases. We demonstrate formally  that our general boundary conditions are the most general possible that are consistent with linearity and causality. Finally, in section \ref{conclusion} we conclude. 
%
%
\section{Boundary conditions in four dimensions}
\label{sec:bcs4d}

%
%
\subsection{`Standard' boundary conditions through integrating Maxwell's equations\label{sec:standard}}
From Maxwell's equations one can always prescribe  natural BCs. Through applying a ``pillbox'' of integration to Maxwell's equations in the vicinity of the boundary---as found in standard textbooks~\cite{jackson1962} (see Fig. \ref{naturalbc}b)---a static spatial boundary $\Sigma$ containing zero surface charge and current implies the usual continuity conditions
\begin{align}
\VE^\I_\parallel\big\vert_\Sigma = \VE^\II_\parallel\big\vert_\Sigma
,\quad
\VB^\I_\perp\big\vert_\Sigma = \VB^\II_\perp\big\vert_\Sigma
,\quad
\VH^\I_\parallel\big\vert_\Sigma = \VH^\II_\parallel\big\vert_\Sigma
,\quad
\VD^\I_\perp\big\vert_\Sigma = \VD^\II_\perp\big\vert_\Sigma\,.
\label{Intr_Nat_Bdd_spatial}
\end{align}
where `$\I$' and `$\II$' indicate the space--time regions either side of the boundary and `$\parallel$' and `$\perp$' indicate the components of the vectors in the plane and out of the plane of the interface, respectively.
%
%
\begin{figure}[!h]
\centering
\includegraphics[width=1.0\linewidth]{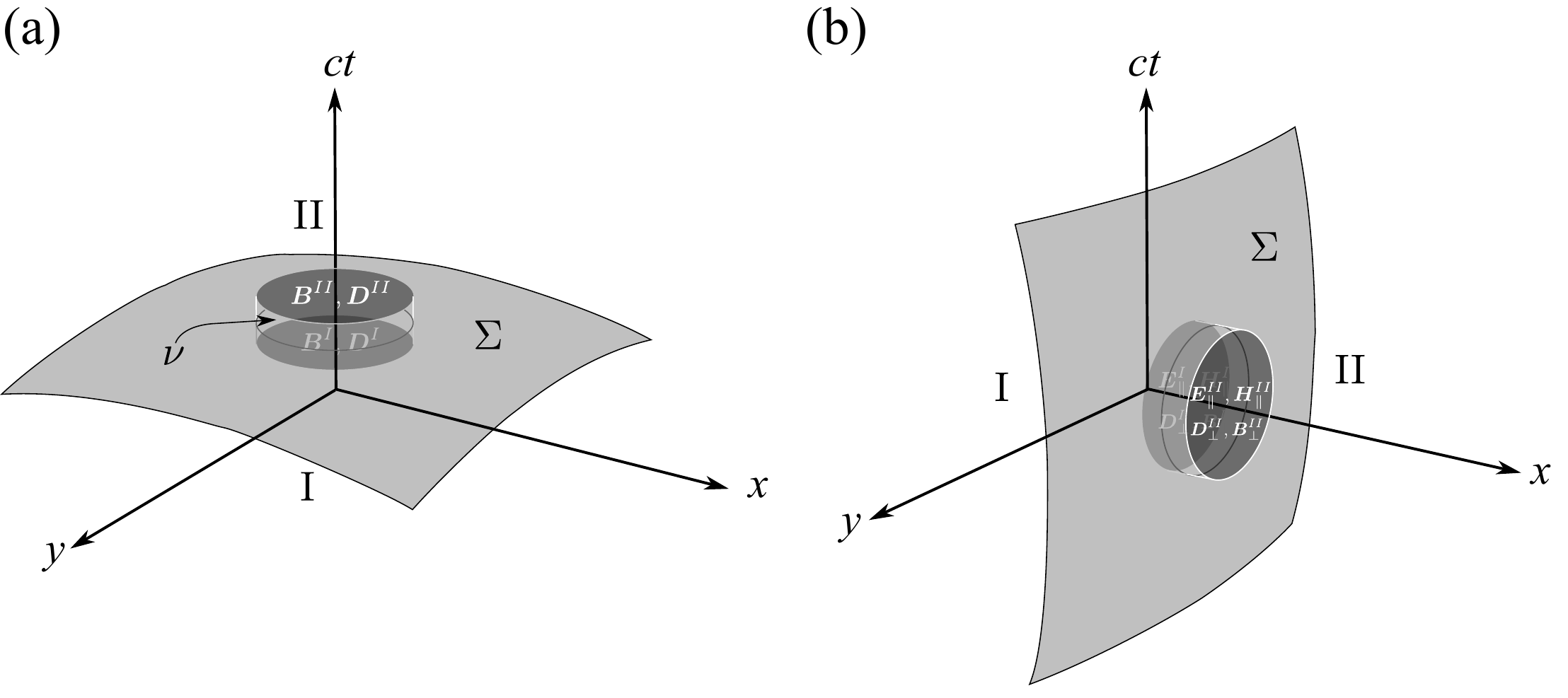}
\caption{\textbf{Spatial and temporal boundaries:} (a) a \emph{temporal} boundary, where the electromagnetic field is related across either side of a spatial hypersurface.  (b) a \emph{spatial} boundary, where the field is related across either side of a surface of constant $x$, resulting in the standard boundary conditions, Eq. \eqref{Intr_Nat_Bdd_spatial}.
There has to be confusion in the language when talking about boundaries. A temporal boundary corresponds to a spatial hypersurface, while a spatial boundary corresponds to a hypersurface, which contains time and 2 dimensions of space. 
\label{naturalbc}}
\end{figure}
%
%
Given that $\VE^\I$ is only defined in region $\I$ and the boundary while $\VE^\II$ is only defined in region $\II$ and the boundary, the only place where both $\VE^\I_\parallel$ and $\VE^\II_\parallel$ are defined is \emph{on} the boundary $\Sigma$. Thus we introduce the simplified notation where we drop the $\big\vert_\Sigma$, and we implicitly mean evaluated on the boundary. Thus \eqref{Intr_Nat_Bdd_spatial} will be written
\begin{align}
\VE^\I_\parallel = \VE^\II_\parallel
,\quad
\VB^\I_\perp = \VB^\II_\perp
,\quad
\VH^\I_\parallel = \VH^\II_\parallel
,\quad
\VD^\I_\perp = \VD^\II_\perp\,.\label{eq:space_boundary}
\end{align}
An example not often found in textbooks is a temporal boundary, as shown in Fig. \ref{naturalbc}a, where, for example the entire spatial volume changes its constitutive parameters at an instant in time.  In this case the relation between the fields either side of this boundary can be found through integrating the Maxwell equations over an infinitely thin pillbox--like space--time region containing the boundary, like that shown in Fig. \ref{naturalbc}a, yielding
\begin{align}
    \int{\rm d}t\int{\rm d}^{3}\boldsymbol{x}\,\boldsymbol{\nabla}\times\boldsymbol{E}=-\int_{\Sigma}{\rm d}^{3}\boldsymbol{x}\left(\boldsymbol{B}^{\II}-\boldsymbol{B}^{\I}\right)=0\nonumber\\
    \int{\rm d}t\int{\rm d}^{3}\boldsymbol{x}\,\boldsymbol{\nabla}\times\boldsymbol{H}=\int_{\Sigma}{\rm d}^{3}\boldsymbol{x}\left(\boldsymbol{D}^{\II}-\boldsymbol{D}^{\I}\right)=0
\end{align}
As the shape of the integration domain on the boundary $\Sigma$ is arbitrary, this implies the continuity of the full magnetic flux density and displacement field vectors,
\begin{align}
\VB^\I = \VB^\II
,\quad
\VD^\I = \VD^\II
\label{Intr_Nat_Bdd_temporal}
\end{align}
These are the same boundary conditions as discussed as in e.g Ref.~\cite{galiffi2022}.  The same approach---integrating over a thin ``pillbox'' enclosing the boundary, $\Sigma$---can be applied to any space--time boundary, whatever the orientation, resulting in intermediate cases where linear combinations of $\boldsymbol{B}$ and $\boldsymbol{E}$, and $\boldsymbol{D}$ and $\boldsymbol{H}$ are continuous across the boundary~\cite{caloz2022generalized}.

The standard BCs, given in Eq. \eqref{Intr_Nat_Bdd_spatial}, can be augmented by including a surface charge and surface current to the last two equations. The nature of surface currents and how they relate to the fields is another topic, outside the scope of this article. They can be prescribed, in which case the BCs will change from being linear homogenous to linear inhomogeneous. By contrast some surface currents are linearly related to the electromagnetic fields on either side of the boundary. In general, these could be absorbed into the effective linear BCs,  for example, the moving polarization density given in Eq. \eqref{eq:sheet-bc}.

%
%
\subsection{Additional boundary conditions and dispersive materials\label{sec:dispersion}}
The above equations, (\ref{eq:space_boundary}) and (\ref{Intr_Nat_Bdd_temporal}) (and their generalization to space--time boundaries) need to be supplemented with additional conditions if the media are temporally or spatially dispersive.  For a spatial boundary between spatially dispersive (sometimes called ``non--local'') media, these extra conditions are known as additional boundary conditions (ABCs), and have been studied since the pioneering work of Pekar~\cite{pekar1957}.  It is less well appreciated that analogous conditions are also required for a temporal boundary between temporally dispersive materials, although this was pointed out in~\cite{gratus2021temporal}.

Taking the aforementioned case of a temporal interface between temporally dispersive media (as in Fig. \ref{naturalbc}a), we assume that both sides of the boundary have a polarization field, $\boldsymbol{P}^\I(t,x)$ and $\boldsymbol{P}^\II(t,x)$, which in each region satisfies the equation of damped simple harmonic motion, e.g.
\begin{align}
\partial_t^2 \boldsymbol{P}^\I + \lambda_\I\partial_t \boldsymbol{P}^\I + \alpha_\I^2 \boldsymbol{P}^\I = \chi_\I \boldsymbol{E}^{\I}
\label{Intr_Temp_disp}
\end{align}
and likewise in the second region, for $\boldsymbol{P}^\II$, with corresponding $\lambda_\II$, $\alpha_\II$ and $\chi_\II$. If---as is the typical case in e.g optics---there were just region $\I$ for all time, one could take the Fourier transform, giving the familiar frequency domain relation
\begin{align}
\tilde{\boldsymbol{P}}^\I(\omega,\boldsymbol{x}) =
\frac{\chi_\I }{-\omega^2 + i\lambda_\I \omega + \alpha_\I^2} \tilde{\boldsymbol{E}}^\I
\label{Intr_Temp_disp_Fourier}
\end{align}
However, this Fourier transform is no longer possible when there is an abrupt change in the material properties.  In order for Eq. (\ref{Intr_Temp_disp}) to be well defined across the boundary $\Sigma$, the polarization and its first derivative must be continuous,
\begin{align}
\boldsymbol{P}^\I = \boldsymbol{P}^\II
\quadand
\partial_t \boldsymbol{P}^\I = \partial_t \boldsymbol{P}^\II\,.
\label{Intr_Nat_Bd_Temp_disp}
\end{align}
For a temporal boundary we must use these conditions in addition to the continuity of the electromagnetic field (\ref{Intr_Nat_Bdd_temporal}) arising from integrating Maxwell's equations over the boundary.  As mentioned above, these additional boundary conditions are the temporal counterpart of those required at a \emph{spatial} boundary between \emph{spatially} dispersive media.  In the case of spatially and temporally dispersive materials, the driven simple harmonic oscillator equation of motion for the polarization field (\ref{Intr_Temp_disp}) is usually supplemented with an additional term containing second order spatial derivatives of the polarization, 
\begin{align}
\partial_t^2 \boldsymbol{P}^\I + \lambda_\I\partial \boldsymbol{P}^\I + \alpha_\I^2 \boldsymbol{P}^\I - \beta_\I^2 \partial_x^2 \boldsymbol{P}^\I= \chi_\I \boldsymbol{E}
\label{Intr_Spatial_Disp}
\end{align}
and likewise for $\boldsymbol{P}^\II$ where $\beta_\II$ is interpreted as the polariton velocity.  This modification results in polarization waves within the material.  In order that Eq. (\ref{Intr_Spatial_Disp}) is well defined we must have the analogous continuity of the polarization and its first spatial derivative,
\begin{align}
\boldsymbol{P}^\I = \boldsymbol{P}^\II
\quadand
\partial_x \boldsymbol{P}^\I = \partial_x \boldsymbol{P}^\II
\label{Intr_Nat_Bd_Spac_disp}
\end{align}
Comparing the non dispersive and dispersive cases for either temporal or spatial boundaries we can see that we must either use Eqns. (\ref{eq:space_boundary}) or (\ref{Intr_Nat_Bdd_temporal}) in isolation,  or---when dispersion is important---supplement these with the additional conditions on the polarization given in Eqns. (\ref{Intr_Spatial_Disp}) or (\ref{Intr_Nat_Bd_Spac_disp}).

There are also examples of spatially dispersive media where no additional boundary conditions are required.  As an example of such a spatiotemporally dispersive boundary condition, take a static thin sheet at $z=0$, within which there is a confined polarization oriented along the $x$--axis, $P_x(t,x,y,z)=\delta(z)\mathcal{P}_x(t,x,y)$.  Within the sheet we assume this in--plane component of the polarization satisfies the following second order differential equation
\begin{align}
\partial_{t}^2 \mathcal{P}_x + 2\lambda_D \partial_t \mathcal{P}_x +
(\omega_0^2\beta^2+\lambda_D^2) \mathcal{P}_x - \beta^2 \partial_x^2 \mathcal{P}_x = \chi_0 E_x
\label{Disp_DE}
\end{align}
Performing a Fourier transform of both sides of Eq. (\ref{Disp_DE}), we can express the inhomogeneous solution to \eqref{Disp_DE} in Fourier space as
\begin{align}
\tilde{\mathcal{P}}_x(\omega,k_x,k_y)
=
\frac{\chi_0 \tilde{E}_x(\omega,k_x,k_y) }
{-\omega^2-2i\lambda_D \omega + (\omega_0^2\beta^2+\lambda_D^2)+\beta^2 k_x^2}
\label{Disp_FT}
\end{align}
which can be understood as a generalization of the usual Lorentz model.  After an inverse Fourier transform of Eq. (\ref{Disp_FT}), we have a convolution in space and time
\begin{align}
\mathcal{P}_x(\tM,\xM,\yM)
=
\int d\tM' \int d\xM' \int d\yM'
\chi(\tM-\tM',\xM-\xM',\yM-\yM') \
E_x(\tM',\xM',\yM')
\label{Disp_Green_Int}
\end{align}
where the kernel of the integral (the in--plane susceptibility) can be found via inverse Fourier transform,
\begin{align}
    \chi(t,x,y)&=\delta(y)\int_{-\infty}^{\infty}\frac{{\rm d}\omega}{2\pi}\int_{-\infty}^{\infty}\frac{{\rm d}k}{2\pi}\frac{{\rm e}^{i[kx-\omega t]}}{-\omega^2-2i\lambda_{D}\omega+(\omega_0^2\beta^2+\lambda_D^2)+\beta^2 k^2}\nonumber\\[10pt]
    &=\frac{1}{2\beta}\delta(y){\rm e}^{-\lambda_{D}t}\Theta(\beta t-|x|)J_0\left(\omega_0\sqrt{\beta^2t^2-x^2}\right)
\end{align}
where we have evaluated the integral using Cauchy's residue theorem and the integral representation of the Hankel function~\cite{gradshteyn}, $J_0$ is the zeroth order Bessel function of the first kind, and $\Theta$ is the Heaviside step function.  The current within the sheet is given by $j_x=\delta(z)\partial_t\mathcal{P}_x$, and using the pillbox integration argument described in Sec. \ref{sec:standard}, the boundary conditions relating the fields in region $\I$ $(z<0)$ and region $\II$ $(z>0)$ are given by the standard boundary conditions
\begin{align}
E^\II_x=E^\I_x
,\quad
E^\II_y=E^\I_y
,\quad
B^\II_z=B^\I_z
,\quad
H^\II_x=H^\I_x
,\quad
H^\II_y=H^\I_y- \partial_t \mathcal{P}_x
\quadand
D^\II_z=D^\I_z
\label{Disp_BC}
\end{align}
(note that in this case there is no need for additional boundary conditions on the polarization field, due to the confinement of the polarization to propagate within the sheet).  Due to the space--time convolution in Eq. (\ref{Disp_Green_Int}) the boundary conditions in (\ref{Disp_BC}) are non--local, but are local to the boundary.  

 It is thus clear that in the general case, the number of boundary conditions is not fixed, and need not be local either. There has to be a sufficient number of boundary conditions for the number of fields, and their derivatives on either side of the boundary. Since we do not want to limit the number of boundary conditions, here we specify simply that there are some number $\dimXi$ of them.  This number increases as the spatial and temporal dispersion becomes more complicated, involving e.g. equations of motion with higher order spatial or temporal derivatives.
%
%
\subsection{Examples of boundary conditions in four--dimensions\label{sec:4d-bc}}
We see that the nature of the above boundary conditions depends not only on the importance of dispersion, but also on the  on the nature of the boundary and it is not immediately obvious how either conditions (\ref{eq:space_boundary}) and (\ref{Intr_Nat_Bdd_temporal}), or (\ref{Intr_Spatial_Disp}) and (\ref{Intr_Nat_Bd_Spac_disp}) should be adapted to the general case.  For instance, these conditions become more complicated if the boundary is moving, as discussed briefly in the introduction.  Through exploiting the spacetime notation of special and general relativity, the above equations can be naturally combined into a single equation, valid for any space--time interface.

As discussed in many texts on classical electromagnetism (see e.g.~\cite{jackson1962} or~\cite{volume2}), the electromagnetic field vectors $\VE$ and $\VB$ are components of a single rank 2 anti-symmetric tensor $F_{\mu\nu}=-F_{\nu\mu}$, where 
\begin{align}
E_i=c\, F_{0i}
\quadand 
B_i=-\tfrac{1}{2}\epsilon_{0ijk}\,\omega\, g^{j\kappa}\, g^{k\lambda}\, F_{\kappa\lambda}
\label{Intr_EB_F}
\end{align}
where 
\begin{align}
    \omega = \sqrt{|\det(g_{\mu\nu})|}
    \label{Intr_measure}
\end{align}
is the measure for the metric tensor $g_{\mu\nu}$. The signature of metric is taken as $(-,+,+,+)$. Note that for Minkowski spacetime with standard Cartesian coordinates, $\omega=1$. As a result, in Minkowski spacetime, $\omega$ also corresponds to the Jacobian associated with the coordinate transformation from Cartesian coordinates to the new coordinates. 
Latin indices range over the spatial coordinates and Greek over the space--time coordinates, and $\epsilon_{\mu\nu\sigma\tau}$ is the unit anti--symmetric tensor, equal to $+1$ for an even permutation of $\{0,1,2,3\}$, $-1$ for an odd permutation, and $0$ otherwise.  Likewise the excitation fields $\VH$ and $\VD$ are encoded in a single anti-symmetric tensor $H_{\mu\nu}$ where
\begin{align}
H_{i}=H_{0i}
\quadand 
D_{i}=\tfrac{1}{2c}\,\epsilon_{0ijk}\, \omega\, g^{j\kappa}\, g^{k\lambda} \,H_{\kappa\lambda}
\label{Intr_DH_H}
\end{align}
Using these two four--dimensional tensors, $F_{\mu\nu}$ and $H_{\mu\nu}$, Maxwell's equations become two simple relations,
\begin{align}
\partial_{[\mu} F_{\nu\rho]} = 0
\quadand
\partial_{[\mu} H_{\nu\rho]} = \epsilon_{\mu\nu\rho\sigma}\, \omega\, J^\sigma
\label{Inro_Max_FH}
\end{align}
where the square brackets indicate an anti-symmetrization with respect to the enclosed space--time indices.  By using this spacetime tensor notation, a single equation can encode boundary conditions that have an arbitrary shape and orientation in space--time, in terms of an arbitrary coordinate system with metric tensor $g_{\mu\nu}$.

We now can re--write the above special cases of space--time boundary conditions---(\ref{eq:space_boundary}) and (\ref{Intr_Nat_Bdd_temporal})---in a general form.  To do this we introduce a set of functions to describe the boundary, $\Sigma^{\mu}(u^1,u^2,u^3)$, depending on three internal coordinates, $u^1$, $u^2$, and $u^3$.  The functions $\Sigma^{\mu}$ are defined such that, for the surface normal $n_\mu$ the equation of the space--time surface is given by $n_{\mu}\Sigma^{\mu}=0$.   The continuity of the in--plane components of the field tensors across any space--time boundary is thus expressed as,
\begin{align}
F^\I_{\mu\nu} \pfrac{\Sigma^\mu}{u^a} \pfrac{\Sigma^\nu}{u^b} 
= 
F^\II_{\mu\nu} \pfrac{\Sigma^\mu}{u^a} \pfrac{\Sigma^\nu}{u^b} 
\quadand
H^\I_{\mu\nu} \pfrac{\Sigma^\mu}{u^a} \pfrac{\Sigma^\nu}{u^b}
=
H^\II_{\mu\nu} \pfrac{\Sigma^\mu}{u^a} \pfrac{\Sigma^\nu}{u^b}
\label{Intr_Nat_Bd_Sigma}.
\end{align}
where again the superscripts `$\I$' and `$\II$' indicate the regions either side of the boundary.  Assuming Cartesian coordinates in Minkowski space--time, a spatial boundary oriented along the $x$--axis corresponds to $\Sigma^{\mu}=(u_1,0,u_2,u_3)$.  In this case Eq. (\ref{Intr_Nat_Bd_Sigma}) reduces to the continuity of $F_{02}$, $F_{03}$, and $F_{23}$ (i.e. $\boldsymbol{E}_{\parallel}$ and $\boldsymbol{B}_{\perp}$) as well as $H_{02}$, $H_{03}$, and $H_{23}$ (i.e. $\boldsymbol{H}_{\parallel}$ and $\boldsymbol{D}_{\perp}$).  These are precisely the boundary conditions prescribed in Eq. (\ref{eq:space_boundary}).  The same is also true for a temporal boundary, where $\Sigma^{\mu}=(0,u_1,u_2,u_3)$, the application of which reduces Eq. (\ref{Intr_Nat_Bd_Sigma}) to the continuity of $F_{ij}$ and $H_{ij}$, i.e. of the vectors $\boldsymbol{B}$ and $\boldsymbol{D}$, as specified by Eq. (\ref{Intr_Nat_Bdd_temporal}). 

The boundary conditions (\ref{Intr_Nat_Bd_Sigma}) can also be applied to moving boundaries, similar to the one dimensional example discussed in the introduction.  For example, again assuming Cartesian coordinates in Minkowski space--time, a surface extended along the $y$ and $z$ axes, and moving along the $x$--axis at velocity $V$ corresponds to $\Sigma^{\mu}=(u_1,V u_1,u_2,u_3)$, which after substitution into the general boundary conditions (\ref{Intr_Nat_Bd_Sigma}), gives
\begin{align}
    F^{\I}_{0,2}+\frac{V}{c}F^{\I}_{1,2}&=F^{\II}_{0,2}+\frac{V}{c}F^{\II}_{1,2},\nonumber\\
    F^{\I}_{0,3}+\frac{V}{c}F^{\I}_{1,3}&=F^{\II}_{0,3}+\frac{V}{c}F^{\II}_{1,3},\nonumber\\\text{and}\quad
    F^{\I}_{23}&=F^{\II}_{23},\label{eq:moving-boundary-1}
\end{align}
equivalent to the continuity of $\boldsymbol{E}_{\parallel}+\boldsymbol{V}\times\boldsymbol{B}_{\parallel}$, and $\boldsymbol{B}_{\perp}$ across the moving interface.  Similarly, for the excitation tensor
\begin{align}
    H^{\I}_{0,2}+\frac{V}{c}H^{\I}_{1,2}&=H^{\II}_{0,2}+\frac{V}{c}H^{\II}_{1,2},\nonumber\\
    H^{\I}_{0,3}+\frac{V}{c}H^{\I}_{1,3}&=H^{\II}_{0,3}+\frac{V}{c}H^{\II}_{1,3},\nonumber\\\text{and}\quad
    H^{\I}_{23}&=H^{\II}_{23},\label{eq:moving-boundary-2}
\end{align}
which is equivalent to the continuity of $\boldsymbol{H}_{\parallel}-\boldsymbol{V}\times\boldsymbol{D}_{\parallel}$ and $\boldsymbol{D}_{\perp}$.  These are the same boundary conditions given in e.g.~\cite{bahrami2024}.

More generally we can also apply this formalism to accelerating boundaries.  For a linear acceleration we can simply write e.g. $\Sigma^{\mu}=(u_1,\int^{u_1}V(t'){\rm d}t',u_2,u_3)$, which gives again (\ref{eq:moving-boundary-1}--\ref{eq:moving-boundary-2}) but with the velocity evaluated at a fixed time.  This is the expected result, applying the conventional boundary conditions (\ref{eq:space_boundary}) in the instantaneous rest frame of the boundary.  However unlike the motion of a real material body---where the velocity is restricted to be less than $c$ for the Lorentz transformations to be valid---here the velocity can exceed that of light (as can be obtained via space--time modulated, `synthetic motion'~\cite{harwood2024}), without any complication arising in the theory.

\subsection{Adapted coordinate systems}
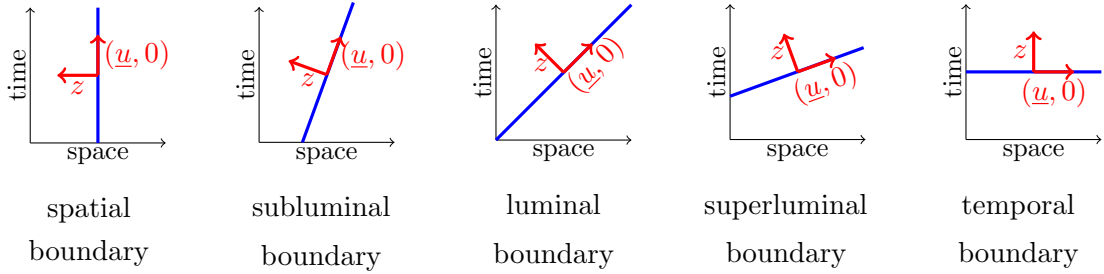
\begin{figure}
\parbox{0.18\textwidth}{\centering
\begin{tikzpicture}[scale=1.8]
\draw [->] (0,0) -- node[below=-2]{\footnotesize space} (1,0) ;
\draw [->] (0,0) -- node[rotate=90,above=-2]{\footnotesize time} (0,1) ;
\draw [very thick,blue] (.5,0) -- +(0,1) ;
\draw [very thick,red,->] (.5,.5) -- node[right=-2]{\small $(\uvec,0)$} +(0,0.3) ;
\draw [very thick,red,->] (.5,.5) -- node[below=-2]{\small $z$} +(-0.3,0) ;
\end{tikzpicture}
\\[-0.5em]
\small spatial \\[-0.5em] boundary}
\parbox{0.18\textwidth}{
\begin{tikzpicture}[scale=1.8]
\draw [->] (0,0) -- node[below=-2]{\footnotesize space} (1,0) ;
\draw [->] (0,0) -- node[rotate=90,above=-2]{\footnotesize time} (0,1) ;
\draw [very thick,blue] (.32,0) -- +(70:1.1) ;
\draw [very thick,red,->] (.5,.5) -- node[right=-2]{\small $(\uvec,0)$} +(70:0.3) ;
\draw [very thick,red,->] (.5,.5) -- node[below=-2]{\small $z$} +({70+90}:0.3) ;
\end{tikzpicture}
\\[-0.5em]
\centering
\small subluminal\\[-0.3em] boundary
}
\parbox{0.18\textwidth}{\centering
\begin{tikzpicture}[scale=1.8]
\draw [->] (0,0) -- node[below=-2]{\footnotesize space} (1,0) ;
\draw [->] (0,0) -- node[rotate=90,above=-2]{\footnotesize time} (0,1) ;
\draw [very thick,blue] (0,0) -- +(1,1) ;
\draw [very thick,red,->] (.5,.5) -- node[below=-4,rotate=45]{\small $(\uvec,0)$} +(45:0.3) ;
\draw [very thick,red,->] (.5,.5) -- node[below left=-4]{\small $z$} +({45+90}:0.3) ;
\end{tikzpicture}
\\[-0.5em]
\small luminal \\[-0.3em] 
boundary
}
\parbox{0.18\textwidth}{\centering
\begin{tikzpicture}[scale=1.8]
\draw [->] (0,0) -- node[below=-2]{\footnotesize space} (1,0) ;
\draw [->] (0,0) -- node[rotate=90,above=-2]{\footnotesize time} (0,1) ;
\draw [very thick,blue] (0,0.32) -- +(20:1.05) ;
\draw [very thick,red,->] (.5,.5) -- node[below=-2,rotate=20]{\small $(\uvec,0)$} +(20:0.3) ;
\draw [very thick,red,->] (.5,.5) -- node[left=-2]{\small $z$} +({20+90}:0.3) ;
\end{tikzpicture}
\\[-0.5em]
\small superluminal \\[-0.3em] 
boundary
}
\parbox{0.18\textwidth}{\centering
\begin{tikzpicture}[scale=1.8]
\draw [->] (0,0) -- node[below=-2]{\footnotesize space} (1,0) ;
\draw [->] (0,0) -- node[rotate=90,above=-2]{\footnotesize time} (0,1) ;
\draw [very thick,blue] (0,0.5) -- +(0:1.0) ;
\draw [very thick,red,->] (.5,.5) -- node[below=-2]{\small $(\uvec,0)$} +(0:0.3) ;
\draw [very thick,red,->] (.5,.5) -- node[left=-2]{\small $z$} +({0+90}:0.3) ;
\end{tikzpicture}
\\[-0.5em]
\small temporal \\[-0.3em] 
boundary
}
\caption{\textbf{Illustrating the benefits of using relativity:} a single set of BCs, for example \eqref{Intr_Nat_Bd_Sigma}, can be applied to arbitrarily moving boundaries (the blue lines indicated above). 
We use a coordinate system, $(u^1,u^2,u^3,z)$ adapted to the boundary, so that $z=0$ corresponds to points on the boundary. We set $\uvec=(u^1,u^2,u^3)$ so that the points $(\uvec,0)$ lie within the boundary, while $z=u^4$ is the transverse coordinate.}
\label{fig:BC-GR}
\end{figure}
One of the key advantages of using tensor calculus and the language of relativity is we can choose a coordinate system adapted to the boundary.  We use coordinates $(u^1,u^2,u^3,u^4)$ where the boundary is given by $u^4=0$. Thus $u^4$ is the transverse coordinate. It is useful to also use the coordinate system $(u^1,u^2,u^3,z)$, where $z=u^4$.  We emphasize that the symbol $z$ does \emph{not} represent the usual Cartesian coordinate, except in the special case of a spatial boundary where the surface normal points along the $z$--axis. This coordinate system is valid whether the boundary is subluminal, luminal or superlumial, as depicted in figure \ref{fig:BC-GR}. 
The boundary functions $\Sigma^\mu(u^1,u^2,u^3)$ are greatly simplified to become
\begin{align}
    \Sigma^1 (u^1,u^2,u^3) = u^1,\,
    \Sigma^2 (u^1,u^2,u^3) = u^2,\,
    \Sigma^3 (u^1,u^2,u^3) = u^3 \text{ and }
    \Sigma^4 (u^1,u^2,u^3) = 0.
\end{align}
In this coordinate system, many BCs are also simplified so that \eqref{Intr_Nat_Bd_Sigma} become
\begin{align}
F^\I_{a b}  
= 
F^\II_{ a b } 
\quadand
H^\I_{a b} 
=
H^\II_{a b}
\label{Intr_Nat_Bd_Adap}.
\end{align}
Recall that the Latin indices are over $a,b=1,2,3$.

\subsection{General differential BCs}

We can---at least in the cases where non--locality is weak enough to be captured in terms of local derivatives of the field---also include dispersion in a four--dimensional boundary condition. 
In this case any additional fields at the boundary only depend only on the values of $F_{\mu\nu}$ and $H_{\mu\nu}$, and their derivatives at that point. In this case we can generalize the examples \eqref{Intr_Nat_Bd_Temp_disp} and \eqref{Intr_Nat_Bdd_spatial} to cases where the media on one or both sides of a space--time boundary are dispersive.  We use the adapted coordinate system $(u^1,u^2,u^3,u^4)$ described above. We write $\uvecF=(u^1,u^2,u^3,u^4)$ as the three surface coordinates $\uvecT=(u^1,u^2,u^3)$, plus $u_4$ which is the coordinate normal to the surface. Hence the point $(\uvecT,0)$ corresponds to a point on the boundary, $\Sigma$. 

The generalized boundary conditions are then given by
\begin{align}
\begin{split}
&
\sum_{\svecF}
\left(
\kappabar^{\F\I\mu\nu}_{\xia \svecF}(\uvecT)
\,
\frac{\partial^\svecF}{\partial\uvecF^\svecF}
F^{\I}_{\mu\nu}(\uvecT,0)
+
\kappabar^{\H\I\mu\nu}_{\xia \svecF}(\uvecT)
\,
\frac{\partial^\svecF}{\partial\uvecF^\svecF}
H^{\I}_{\mu\nu}(\uvecT,0)
\right)
\\&
\qquad =
\sum_{\svecF}
\left(
\kappabar^{\F\II\mu\nu}_{\xia \svecF}(\uvecT)
\,
\frac{\partial^\svecF}{\partial\uvecF^\svecF}
 F^{\II}_{\mu\nu}(\uvecT,0)
+
\kappabar^{\H\II\mu\nu}_{\xia \svecF}(\uvecT)
\,
\frac{\partial^\svecF}{\partial\uvecF^\svecF}
 H^{\II}_{\mu\nu}(\uvecT,0)
\right)
\end{split}
\label{BddPt_Psi_action}
\end{align}
where $A=1,\ldots,\dimXi$ labels each of the BCs, the number of which is determined by the equations of motion of the medium in regions $\I$ and $\II$.
In Eq. (\ref{BddPt_Psi_action}), the vector $\svecF=(s_1,s_2,s_3,s_4)$ contains the orders of the derivatives of each of the four coordinates to be evaluated on the boundary.

Note that the sums $\sum_{\svecF}$ range over a finite number of combinations of space--time derivatives. Written out explicitly, this notation implies a product of derivatives of order $s_i$ with respect to each of the coordinates,
\begin{align}
\frac{\partial^\svecF}{\partial\uvecF^\svecF}
F^{\I}_{\mu\nu}(\uvecT,0)
=
\Big(\pfrac{}{u^1}\Big)^{s_1}
\Big(\pfrac{}{u^2}\Big)^{s_2}
\Big(\pfrac{}{u^3}\Big)^{s_3}
\Big(\pfrac{}{u^4}\Big)^{s_4}
F^{\I}_{\mu\nu}(\uvecT,0)\label{eq:d_us_definition}
\end{align}
Note that the derivative $\partial/\partial u^4$ differentiates normal to the boundary. Since the fields $\Set{F^\I_{\mu\nu},H^\I_{\mu\nu},F^\II_{\mu\nu},H^\II_{\mu\nu}}$ are only defined up to the boundary these must be one-sided derivatives. The other three derivatives, differentiate inside the boundary and are two-sided derivatives.
%
%
\section{Generalized Linear Boundary Conditions}
\label{Sec:GenBdd Formula}

The discussion so far has highlighted a large range of possible (linear) boundary conditions.  For standard spatial boundaries, these can be encompassed in the usual field continuity conditions (\ref{Intr_Nat_Bdd_spatial}) (with appropriate modifications to allow for surface charge and current), plus the additional conditions due to spatial and temporal dispersion described in Sec. \ref{sec:dispersion}.  However, it is clear from the preceding section \ref{sec:4d-bc} that there is a much wider class of boundary conditions at a general space--time interface.  In this section we present boundary conditions that encompass all the above examples. We will then proceed in Sec. \ref{Sec:GenBddAxioms}  to show that these are the most general boundary conditions that satisfy a set of physically reasonable conditions.

We choose a coordinate system $(u^1,u^2,u^3,z)$ which is adapted to the boundary.  In the general case the boundary conditions can be written in terms of four integral kernels, $\Set{\kappa^{\F^\I}_{Ar}{}^{\mu\nu}(\uvec,\uvec'),\kappa^{\H\I}_{Ar}{}^{\mu\nu}(\uvec,\uvec'),\kappa^{\F\II}_{Ar}{}^{\mu\nu}(\uvec,\uvec'),\kappa^{\H\II}_{Ar}{}^{\mu\nu}(\uvec,\uvec')}$.  These kernels must be distributions in $\uvec'$ to allow for the $\delta$--functions required for local boundary conditions. The kernels are also assumed antisymmetric $\kappa^{\F\I}_{Ar}{}^{\mu\nu}(\uvec,\uvec')=-\kappa^{\F\I}_{Ar}{}^{\nu\mu}(\uvec,\uvec')$. 
In terms of these integral kernels the general boundary conditions are then given by

\begin{align}
\begin{split}
&\tfrac12\sum_{r=0}^k\int_{J^\Sigma(\uvec)} 
\Big(
\kappa^{\F\I }_{Ar}{}^{\mu\nu}(\uvec,\uvec')
\,
\frac{\partial^r }{\partial z^r}F^\I_{\mu\nu}(\uvec',0)
+
\kappa^{\H\I }_{Ar}{}^{\mu\nu}(\uvec,\uvec')
\,
\frac{\partial^r }{\partial z^r}H^\I_{\mu\nu}(\uvec',0)
\Big)
\,\omega_{(\uvec,z)}\, d^3\uvec'
\\&
\qquad
=
\tfrac12\sum_{r=0}^k\int_{J^\Sigma(\uvec)} 
\Big(
\kappa^{\F\II }_{Ar}{}^{\mu\nu}(\uvec,\uvec')
\,
\frac{\partial^r }{\partial z^r}F^\II_{\mu\nu}(\uvec',0)
+
\kappa^{\H\II }_{Ar}{}^{\mu\nu}(\uvec,\uvec')
\,
\frac{\partial^r }{\partial z^r}H^\II_{\mu\nu}(\uvec',0)
\Big)
\,\omega_{(\uvec,z)}\,d^3\uvec'.
\end{split}
\label{Bdd_Gen_Bdd}
\end{align}
where $A=1,\ldots,\dimXi$ again label the different boundary conditions, and $\omega_{(\uvec,z)}$ is the measure \eqref{Intr_measure} expressed in the $(\uvec,z)$ coordinate system.

Having stated \eqref{Bdd_Gen_Bdd} as the most general space--time boundary condition, we show here that these include all of the examples we have presented up to now.  In Minkowski space, the natural boundary conditions \eqref{Intr_Nat_Bd_Sigma}, which are the space--time generalisation of the standard conditions, \eqref{Intr_Nat_Bdd_spatial} and \eqref{Intr_Nat_Bdd_temporal}, arise by setting $m=6$ and $k=0$.  For $A=1,2,3$ we can give the condition for the continuity of the tangential fields as
\begin{align}
\kappa^{\F\I }_{A0}{}^{\mu\nu}(\uvec,\uvec') &= 
\kappa^{\F\II }_{A0}{}^{\mu\nu}(\uvec,\uvec') = 
\tfrac12\frac{\epsilon^{abc}}{\omega_{(\uvec,z)}}\delta^{\mu}_{b}\delta^{\nu}_{c}\delta(\uvec-\uvec'),
\nonumber\\[10pt]
\kappa^{\H\I }_{A0}{}^{\mu\nu}(\uvec,\uvec') &=
\kappa^{\H\II }_{A0}{}^{\mu\nu}(\uvec,\uvec') = 0
\label{GenBC_Std_BC}
\end{align}
where $a=A$, $\epsilon^{abc}$ is the antisymmetric unit tensor.  Likewise for $A=4,5,6$ we have the similar relation
\begin{align}
\kappa^{\F\I }_{A0}{}^{\mu\nu}(\uvec,\uvec') &=
\kappa^{\F\II }_{A0}{}^{\mu\nu}(\uvec,\uvec') = 0\nonumber\\[10pt]
\kappa^{\H\I }_{A0}{}^{\mu\nu}(\uvec,\uvec') &= 
\kappa^{\H\II }_{A0}{}^{\mu\nu}(\uvec,\uvec') = 
\frac{\epsilon^{abc}}{2\omega_{(\uvec,z)}}\delta^{\mu}_{b}\delta^{\nu}_c\delta(\uvec-\uvec')
\label{GenBC_Std_BC_H}
\end{align}
where $a=A-3$.  Note that -- as described above -- here we are using a coordinate system where three of the four coordinates follow the boundary, plus one normal coordinate.  The four dimensional indices $\mu$ and $\nu$ in Eq. (\ref{GenBC_Std_BC_H}) are in this coordinate system, which simplifies Eqns. (\ref{GenBC_Std_BC}--\ref{GenBC_Std_BC_H}) compared to the earlier expressions (\ref{Intr_Nat_Bd_Sigma}).

As discussed above, when dealing with dispersive media we may need to add additional boundary conditions to those given in (\ref{GenBC_Std_BC} and (\ref{GenBC_Std_BC_H}).  These can also be encoded in the general expression (\ref{Bdd_Gen_Bdd}).  For example, using the definition of the displacement field (\ref{Intr_DH_H}) in Minkowski space covered with Cartesian coordinates, the polarization field equals e.g. $P^\I_x=c^{-1}H^\I_{23}-\epsilon_0 c F^\I_{01}$. Thus the continuity of the $x$--component of the polarization $P_x^\I = P_x^\II$ can be included in Eq. (\ref{Bdd_Gen_Bdd}) (here $A=7$)
\begin{align}
\begin{aligned}
&\kappa^{\F\I }_{70}{}^{\mu\nu}(\uvec,\uvec') = 
\kappa^{\F\II }_{70}{}^{\mu\nu}(\uvec,\uvec') = 
-\tfrac{\epsilon_0 c}{2}(\delta^\mu_0\delta^\nu_1-\delta^\nu_0\delta^\mu_1) \ \delta(\uvec-\uvec')
\quadand
\\
&\kappa^{\H\I }_{70}{}^{\mu\nu}(\uvec,\uvec') = 
\kappa^{\H\II }_{70}{}^{\mu\nu}(\uvec,\uvec') = 
\tfrac{1}{2c}(\delta^\mu_2\delta^\nu_3-\delta^\nu_2\delta^\mu_3) \ \delta(\uvec-\uvec').
\end{aligned}
\end{align}
If, in addition the time derivative of the polarization is continuous, $\partial_t P_x^\I = \partial_tP_x^\II$ this boundary equation involves a derivative of the field tensors and we have an increased order of the boundary condition, $k=1$. This equation is given by (when $A=8$)
\begin{align}
\begin{aligned}
&\kappa^{\F\I }_{81}{}^{\mu\nu}(\uvec,\uvec') = 
\kappa^{\F\II }_{81}{}^{\mu\nu}(\uvec,\uvec') = 
-\tfrac{\epsilon c}{2}(\delta^\mu_0\delta^\nu_1-\delta^\nu_0\delta^\mu_1) \ \delta(\uvec-\uvec')
\quadand
\\
&\kappa^{\H\I }_{81}{}^{\mu\nu}(\uvec,\uvec') = 
\kappa^{\H\II }_{81}{}^{\mu\nu}(\uvec,\uvec') = 
\tfrac{1}{2c}(\delta^\mu_2\delta^\nu_3-\delta^\nu_2\delta^\mu_3) \ \delta(\uvec-\uvec')
\end{aligned}
\end{align}
There are also cases where we must implement sheet transition conditions, where the boundary conditions contain mixtures of the electromagnetic field evaluated either side of the boundary, and the internal polarization and/or magnetization fields.  We gave such an example above, with the boundary conditions (\ref{Disp_BC}) for a metasurface containing a surface current density, $\partial_t\mathcal{P}_x$, where $P_x=\mathcal{P}_x\delta(z)$.  As discussed above in this case there is no need for additional boundary conditions on the polarization field, the dynamics of the polarization entering through the integral relation (\ref{Disp_Green_Int}).  To describe this example using Eq. (\ref{Bdd_Gen_Bdd}), we choose the coordinate system $(u^1,u^2,u^3,u^4)=(x,y,t,z)$, so that the boundary is at $u^4=z=0$.  The boundary conditions  (\ref{Disp_BC}) are mostly those given in  (\ref{GenBC_Std_BC}) and (\ref{GenBC_Std_BC_H}), except for $H_{y}^{\II}=H_{y}^{\I}-\partial_t\mathcal{P}_x$, which is equivalent to the non--local condition
\begin{equation}
    H_{y}^{\II}=H_{y}^{\I}-\int dt'\int dx'\int dy'\partial_{t}\chi(t-t',x-x',y-y')E_x^{\II}(t',x',y')
\end{equation}
which, in terms of the general condition (\ref{Bdd_Gen_Bdd}) requires the following modification to the standard boundary conditions (for the condition numbered $A=4$),
\begin{align}
\kappa^{\H\I 31}_{40}(\uvec,\uvec') &= \delta(\uvec-\uvec'),\nonumber\\
\kappa^{\H\II 31}_{40}(\uvec,\uvec') &= \delta(\uvec-\uvec'),\nonumber\quad
\text{and}\quad\\
\kappa^{\F\II 31}_{40}(\uvec,\uvec') &= c\frac{\partial}{\partial u^3}\chi(u^3-u'{}^3,u^1-u'{}^1,u^2-u'{}^{2}),
\end{align}
Dispersive boundary effects can---as also discussed above---be approximated through a series expansion of the non--local relations (\ref{Bdd_Gen_Bdd}), as given in Eq. (\ref{BddPt_Psi_action}).  To reproduce this family of boundary conditions we simply set the integral kernels to equal a set of derivatives of delta functions 
\begin{align}
\kappa^{\F\I\mu\nu}_{A {s_4}}(\uvec,\uvec')
=
\sum_{s_1,s_2,s_3}
\frac{(-1)^{s_1+s_2+s_3}}{\omega_{(\uvec,z)}}
\kappabar^{\F\I\mu\nu}_{A \svecF}\, \frac{\partial^\svecF}{\partial\uvecF^\svecF}
\delta(\uvec-\uvec')
\end{align}
and likewise for $\kappa^{\H\I\mu\nu}_{A r}(\uvec,\uvec'), \kappa^{\F\II\mu\nu}_{A r}(\uvec,\uvec'), \kappa^{\H\II\mu\nu}_{A r}(\uvec,\uvec')$ and  $r=s_1+\cdots+s_4$.  The notation is the same as that introduced in Eq. (\ref{eq:d_us_definition}).


\subsection{Extracting the order, $k$ of the boundary conditions, and the integral kernels $\kappa$}
\label{SubSec:Extrating k and kappa}

Taking an arbitrary electromagnetic field, the boundary conditions \eqref{Bdd_Gen_Bdd}, will not be satisfied.  To quantify this we define a set of distributions $R_1,\ldots,R_\dimXi$, which represent the remainder of the difference between the two sides of Eq. (\ref{Bdd_Gen_Bdd})
\begin{align}
\begin{split}
\lefteqn{R_A(\uvec)[F^\I_{\mu\nu},H^\I_{\mu\nu},F^\II_{\mu\nu},H^\II_{\mu\nu}]} 
\hspace{3em}&
\\&=
\tfrac12\sum_{r=0}^k\int_{J^\Sigma(\uvec)} 
\Big(
\kappa^{\F\I }_{Ar}{}^{\mu\nu}(\uvec,\uvec')
\,
\frac{\partial^r }{\partial z^r}F^\I_{\mu\nu}
+
\kappa^{\H\I }_{Ar}{}^{\mu\nu}(\uvec,\uvec')
\,
\frac{\partial^r }{\partial z^r}H^\I_{\mu\nu}
\Big)
\,d^3\uvec'
\\&\quad
-
\tfrac12\sum_{r=0}^k\int_{J^\Sigma(\uvec)} 
\Big(
\kappa^{\F\II }_{Ar}{}^{\mu\nu}(\uvec,\uvec')
\,
\frac{\partial^r }{\partial z^r}F^\II_{\mu\nu}
+
\kappa^{\H\II }_{Ar}{}^{\mu\nu}(\uvec,\uvec')
\,
\frac{\partial^r }{\partial z^r}H^\II_{\mu\nu}
\Big)
\,d^3\uvec'
\end{split}
\label{Bdd_Gen_Def_R}
\end{align}
where for brevity we have suppressed the arguments of the field tensors.  Clearly the BCs are satisfied if and only if all the $R_A(\uvec)=0$ for all $\uvec$.  

In order to calculate the order $k$ of the distributions, we replace the electromagnetic fields with $\{z^{k+1}F^\I_{\mu\nu},z^{k+1}H^\I_{\mu\nu},z^{k+1}F^\II_{\mu\nu},z^{k+1}H^\II_{\mu\nu}\}$ then
\begin{align}
\begin{split}
\lefteqn{R_A(\uvec)[z^{k+1}F^\I_{\mu\nu},z^{k+1}H^\I_{\mu\nu},z^{k+1}F^\II_{\mu\nu},z^{k+1}H^\II_{\mu\nu}]} 
\hspace{3em} &
\\
&=
\tfrac12\sum_{r=0}^k\int_{J^\Sigma(\uvec)} 
\Big(
\kappa^{\F\I }_{Ar}{}^{\mu\nu}(\uvec,\uvec')
\,
\frac{\partial^r }{\partial z^r}\big(z^{k+1}F^\I_{\mu\nu}\big)
+
\kappa^{\H\I }_{Ar}{}^{\mu\nu}(\uvec,\uvec')
\,
\frac{\partial^r }{\partial z^r}\big(z^{k+1}H^\I_{\mu\nu}\big)
\Big)
\,d^3\uvec'
\\&\quad
-
\tfrac12\sum_{r=0}^k\int_{J^\Sigma(\uvec)} 
\Big(
\kappa^{\F\II }_{Ar}{}^{\mu\nu}(\uvec,\uvec')
\,
\frac{\partial^r }{\partial z^r}\big(z^{k+1}F^\II_{\mu\nu}\big)
+
\kappa^{\H\II }_{Ar}{}^{\mu\nu}(\uvec,\uvec')
\,
\frac{\partial^r }{\partial z^r}\big(z^{k+1}H^\II_{\mu\nu}\big)
\Big)
\,d^3\uvec'\\
&=0
\end{split}
\label{Bdd_Gen_R_order}
\end{align}
since $z=0$ on $\Sigma$.  Hence we can define the order of the boundary conditions  as the smallest $k$ such that $R_A(\uvec)=0$ for all $A$ and $\uvec$ when applied to $\left\{z^{k+1}F^\I_{\mu\nu},z^{k+1}H^\I_{\mu\nu},z^{k+1}F^\II_{\mu\nu},z^{k+1}H^\II_{\mu\nu}\right\}$.  

We now illustrate how the $\kappa$ tensors appearing in the boundary conditions (\ref{Bdd_Gen_Bdd}) can be extracted from the remainder functions $R_{A}$. This method only works if the $\kappa$ are regular functions and not distributions. 
Let $\psi_0(\sigma)$ be a bump function, so that $\psi_0(0)=1$, $\partial_\sigma^r\psi_0(0)=0$ for all $r\ge1$ and $\int_{-\infty}^\infty\psi_0(\sigma) d\sigma=1$. 
This has the property that $\partial_z^s \big(z^r \psi_0(z)\big)=r!\delta^s_r$ and $\lim_{\epsilon\to 0} \epsilon^{-1} \psi_0(\epsilon^{-1}\sigma)=\delta(\sigma)$.

Now consider the field 
\begin{align}
   T_{\mu\nu}^{\sigma\rho}(r,\uvec,z,\vvec,\epsilon)
   &=
   \frac{1}{2\epsilon^3 r!} (\delta^\sigma_\mu \delta^\rho_\nu -\delta^\sigma_\nu \delta^\rho_\mu)
   \psi_0\Big(\frac{u^1-v^1}{\epsilon}\Big)
   \psi_0\Big(\frac{u^2-v^2}{\epsilon}\Big)
   \psi_0\Big(\frac{u^3-v^3}{\epsilon}\Big)
   z^r\psi_0(z)
\label{Bdd_Gen_Test}
\end{align}
for $\epsilon>0$. Then
\begin{align}
    \kappa^{\F\I }_{Ar}{}^{\sigma\rho}(\uvec,\vvec)
    = \lim_{\epsilon\to0} \Big(
    R_A(\uvec)[T^{\sigma\rho}_{\mu\nu}(r,\uvec,z,\vvec,\epsilon),0,0,0]
    \Big)\,,
    \label{Bdd_Gen_kappa_extact}
\end{align}
and likewise for the other $\kappa$'s.
This follows since
\begin{align}
\begin{split}
&
\lim_{\epsilon\to0}
    R_A(\uvec)[T^{\sigma\rho}_{\mu\nu}(r,\uvec,z,\vvec,\epsilon),0,0,0]
    \\&=
\tfrac12 \lim_{\epsilon\to0}\sum_{s=0}^k\int_{J^\Sigma(\uvec)} 
\kappa^{\F\I }_{As}{}^{\mu\nu}(\uvec,\uvec')
\,
\frac{\partial^s }{\partial z^s}\Big(
T_{\mu\nu}^{\sigma\rho}(r,\uvec',z,\vvec,\epsilon)
\Big)
\,d^3\uvec'
    \\&=
 \lim_{\epsilon\to0}
\frac1{4 r!\epsilon^3}\sum_{s=0}^k\int_{J^\Sigma(\uvec)} 
\kappa^{\F\I }_{As}{}^{\mu\nu}(\uvec,\uvec')
\,
\frac{\partial^s }{\partial z^s}\Big(
(\delta^\sigma_\mu \delta^\rho_\nu -\delta^\sigma_\nu \delta^\rho_\mu)
   \psi_0\Big(\frac{u^1{-}v^1}{\epsilon}\Big)
   \psi_0\Big(\frac{u^2{-}v^2}{\epsilon}\Big)
\\&\hspace{10em}
   \psi_0\Big(\frac{u^3{-}v^3}{\epsilon}\Big)
   z^r\psi_0(z)
\Big)
\,d^3\uvec'
   \\&=
 \lim_{\epsilon\to0}
\frac1{2 r!\epsilon^3}\sum_{s=0}^k\int_{J^\Sigma(\uvec)} 
\kappa^{\F\I }_{As}{}^{\sigma\rho}(\uvec,\uvec')
\,
\frac{\partial^s }{\partial z^s}\Big(
   \psi_0\Big(\frac{u^1{-}v^1}{\epsilon}\Big)
   \psi_0\Big(\frac{u^2{-}v^2}{\epsilon}\Big)
   \psi_0\Big(\frac{u^3{-}v^3}{\epsilon}\Big)
   z^r\psi_0(z)
\Big)
\,d^3\uvec'
   \\&=
 \lim_{\epsilon\to0}
\frac1{2 \epsilon^3}\int_{J^\Sigma(\uvec)} 
\kappa^{\F\I }_{Ar}{}^{\sigma\rho}(\uvec,\uvec')
\,
   \psi_0\Big(\frac{u^1{-}v^1}{\epsilon}\Big)
   \psi_0\Big(\frac{u^2{-}v^2}{\epsilon}\Big)
   \psi_0\Big(\frac{u^3{-}v^3}{\epsilon}\Big)
\,d^3\uvec'
\\&=
\kappa^{\F\I }_{Ar}{}^{\sigma\rho}(\uvec,\vvec)
\end{split}
\end{align}
As stated this assumes the $\kappa$'s are regular functions. For distributional $\kappa$'s the technique will depend on the nature and support of the distibutions.


\subsection{Why are these  the most general BCs?}
\label{Sec:GenBddAxioms}

In section \ref{Sec:GenBdd Formula} we gave a set of BCs. We claimed were the most general BCs which satisfy a set of conditions. In this section we give these conditions and show that \eqref{Bdd_Gen_Bdd} are indeed the most general BCs that satisfy these conditions. The necessary conditions for our boundary conditions to be general are;
\begin{itemize}
\item Linearity.
\item Causality.
\item Dependence only on the fields on the boundary and their derivatives.
\item The boundary conditions are given by a finite set of distributions, $\Set{R_1(\uvec),\ldots R_\dimXi(\uvec)}$ which have a maximum order.
\end{itemize}

The statement that the BCs are linear can be expressed as follows. Given two sets of fields $\Set{F^\I_{\mu\nu},H^\I_{\mu\nu},F^\II_{\mu\nu},H^\II_{\mu\nu}}$ and $\Set{\hat{F}^\I_{\mu\nu},\hat{H}^\I_{\mu\nu},\hat{F}^\II_{\mu\nu},\hat{H}^\II_{\mu\nu}}$, which both satisfy the BC. 
Then we can add them together to create a new set of field which also satisfy the BCs, that is,  $\Set{F^\I_{\mu\nu}+\hat{F}^\I_{\mu\nu},H^\I_{\mu\nu}+\hat{H}^\I_{\mu\nu},F^\II_{\mu\nu}+\hat{F}^\II_{\mu\nu},H^\II_{\mu\nu}+\hat{H}^\II_{\mu\nu}}$ which also satisfy the BC. We can also scale the fields by an arbitrary constant $\lambda$, and they satisfy the BCs, so that  $\Set{\lambda F^\I_{\mu\nu}, \lambda H^\I_{\mu\nu}, \lambda F^\II_{\mu\nu},\lambda H^\II_{\mu\nu}}$ satisfy the BC. We also assume the BC are satisfies if all fields are zero. This last statement is a trivial implication of linearity if we assume that there are at least one solution to the BC.

As stated above, we demand that the BCs themselves are causal. That means that information about the field on one side of the boundary does not travel faster than the speed of light to affect the fields on the other side of the boundary. 

We also demand that the BCs be local to the
boundary. Otherwise they are not really BCs at all. For example one would not expect the value of fields a significant distance from the boundary to be affected directly by the boundary. Thus the BCs should only depend on the value of electromagnetic fields on the boundary and a finite number of derivatives. There is a relationship between the idea of locality and the idea of dispersion. The nature of this dependency may be divided into three categories. Consider a cause $\alpha$ and response $\beta$, both of which are fields on $\MMan$. 
Let  $\alpha|_\ppt$ represent the value of the field $\alpha$ at the point $\ppt\in\MMan$. Then we have:

\begin{itemize}
\item
\textbf{Pointwise dependence}: In this case $\alpha|_\ppt$ depends only on the
value of $\beta|_\ppt$. For example a non dispersive electromagnetic 
constitutive relation $H_{\mu\nu}=\tfrac{1}{2}\kappa_{\mu\nu}{}^{\rho\sigma}\,F_{\rho\sigma}$, where $\kappa_{\mu\nu}{}^{\rho\sigma}$ is a tensor, is pointwise local. 
\item
\textbf{Differential dependence}: In this case $\beta|_\ppt$ depends not
only on the value of $\alpha|_\ppt$ but also on a finite number of
derivatives. Many authors model spatial dispersion in this manner \cite{khrabustovskyi2017interface}
For example one may consider $H_{\mu\nu}$
to depend on a number of derivatives of $F_{\mu\nu}$. Thus to evaluate
$H_{\mu\nu}|_\ppt$, one needs to know the value of $F_{\mu\nu}$, not only at $\ppt$ but
also in an arbitrary small neighbourhood of $\ppt$. 
\item
\textbf{Extended dependence}: In this case $\beta|_\ppt$ depends on values of $\alpha$ over an extended spacetime region. An example of this is the Lorentz model of temporal dispersion. In this case $\VD$ at time $t$ requires integrating $\VE$ in time.
\end{itemize}

For this reason we consider BCs which are pointwise
dependent or differential dependent with respect to the transverse
direction, but may be any of the above, including extended dependence
within the boundary. 

There is some disagreement in the literature as to whether weak spatial dispersion corresponds to pointwise dependent or
differentiably dependent. If one has a medium which is magnetoelectric, including the case of circularly polarised, then the question of whether this is pointwise local or differentiably local depends on how one expresses the constitutive parameters. If one uses $F_{\mu\nu}$ and $H_{\mu\nu}$ then once can relate them using a tensor which included spatial dispersion, thus the constitutive relation is pointwise local. Alternatively if one only uses $\VE$ and $\VH$ then one must use Maxwell's equations to construct the corresponding $\VB$ and $\VD$, which are differentiated once. In this case one has a differentiably dependent relationship. In this article we use $F_{\mu\nu}$ and $H_{\mu\nu}$ and would describe the relationship as pointwise.

We say there is a maximal order if there is some positive integer $k$ such that for all $\Set{z^{k+1}F^\I_{\mu\nu},z^{k+1}H^\I_{\mu\nu},z^{k+1}F^\II_{\mu\nu},z^{k+1}H^\II_{\mu\nu}}$ then the BCs are satisfied.

We are now in a position to show that these are the most general BC
\begin{theorem}
\label{thm_BC}
Given a set of BCs which are linear, causal, local to $\Sigma$ with differential dependence and a maximal order, we assume there exist a set of distributions $\Set{D_1(\uvec),\ldots D_\dimXi(\uvec)}$ for each point $\uvec\in\Sigma$ which have as arguments $[F^\I_{\mu\nu},H^\I_{\mu\nu},F^\II_{\mu\nu},H^\II_{\mu\nu}]$ such that the BCs are given by
\begin{align}
    D_A(\uvec)[F^\I_{\mu\nu},H^\I_{\mu\nu},F^\II_{\mu\nu},H^\II_{\mu\nu}]
    =0
    \label{Bdd_DA},
\end{align}
for all $A=1,\ldots\dimXi$ and $\uvec\in\Sigma$.
Then these BCs are given by equation \eqref{Bdd_Gen_Bdd}.
\end{theorem}

\begin{proof}
In the following we assume we are given the $D_A(\uvec)$ only. 
We can extract the order of the BCs, as the minimum value of $k$ such that
\begin{align*}
    D_A(\uvec)[z^{k+1}F^\I_{\mu\nu},z^{k+1}H^\I_{\mu\nu},z^{k+1}F^\II_{\mu\nu},z^{k+1}H^\II_{\mu\nu}]
    =0
\end{align*}
for all $A$, and all $\Set{F^\I_{\mu\nu},H^\I_{\mu\nu},F^\II_{\mu\nu},H^\II_{\mu\nu}}$.
From linearity we have
\begin{align*}
    D_A(\uvec)[F^\I_{\mu\nu},H^\I_{\mu\nu},F^\II_{\mu\nu},H^\II_{\mu\nu}]
    =
    D^{F\I}_A(\uvec)[F^\I_{\mu\nu}] +
    D^{H\I}_A(\uvec)[H^\I_{\mu\nu}] -
    D^{F\II}_A(\uvec)[F^\II_{\mu\nu}] -
    D^{H\II}_A(\uvec)[H^\II_{\mu\nu}]
\end{align*}
where $D^{F\I}_A(\uvec)[F^\I_{\mu\nu}]=D_A(\uvec)[F^\I_{\mu\nu},0,0,0]$,
$D^{F\II}_A(\uvec)[F^\I_{\mu\nu}]=-D_A(\uvec)[0,0,F^\II_{\mu\nu},0]$, etc. We can Taylor expand 
\begin{align*}
    F^\I_{\mu\nu}
    =
    \sum_{r=0}^k \frac{z^r}{r!}\frac{\partial^r F^\I_{\mu\nu}}{\partial z^r}  \Big\vert_{z=0}
    + z^{k+1}{\cal E}_{\mu\nu}
\end{align*}
where $z^{k+1}{\cal E}_{\mu\nu}$ is the remainder. 
Now $D^{F\I}_A(\uvec)[z^{k+1}{\cal E}_{\mu\nu}]=0$ so we have
\begin{align*}
    D^{F\I}_A(\uvec)[F^\I_{\mu\nu}]
    =
    \sum_{r=0}^k \frac{1}{r!}D^{F\I}_A(\uvec)\Big[z^r \frac{\partial^r F^\I_{\mu\nu}}{\partial z^r}  \Big\vert_{z=0}\Big]
    =
    \sum_{r=0}^k D^{F\I}_{Ar}(\uvec)\Big[\frac{\partial^r F^\I_{\mu\nu}}{\partial z^r}  \Big\vert_{z=0}\Big]
\end{align*}
where the distribution $D^{F\I}_{Ar}(\uvec)$, is given by
\begin{align*}
    D^{F\I}_{Ar}(\uvec)[\alpha_{\mu\nu}]
    =
    \frac{1}{r!} D^{F\I}_A(\uvec)\Big[z^r \alpha_{\mu\nu}\Big] 
\end{align*}
where $\alpha_{\mu\nu}$ is independent of $z$.
Since each $D^{F\I}_{Ar}(\uvec)$ is a distribution only on the 3--dimensional hypersurface $\Sigma$ 
corresponding to $z=0$, we can always write this in terms of a kernel $\kappa^{\F\I\mu\nu}_{Ar}(\uvec,\uvec')$ so that
\begin{align*}
    D^{F\I}_{Ar}(\uvec)[\alpha_{\mu\nu}]
    =
    \tfrac12 \int_\Sigma
    \kappa^{\F\I }_{Ar}{}^{\mu\nu}(\uvec,\uvec')
    \,
    \alpha_{\mu\nu}
    \,d^3\uvec'
\end{align*}
However from causality we have to restrict the domain of the integral to $J^\Sigma(\uvec)$, giving
\begin{align*}
    D^{F\I}_{Ar}(\uvec)[\alpha_{\mu\nu}]
    =
    \tfrac12 \int_{J^\Sigma(\uvec)}
    \kappa^{\F\I }_{Ar}{}^{\mu\nu}(\uvec,\uvec')
    \,
    \alpha_{\mu\nu}
    \,d^3\uvec'
\end{align*}
Thus we see that $D_A(\uvec)[F^\I_{\mu\nu},H^\I_{\mu\nu},F^\II_{\mu\nu},H^\II_{\mu\nu}]$, is given by the right hand side of \eqref{Bdd_Gen_Def_R}.

\end{proof}

\section{Summary and conclusions}\label{conclusion}

In the above we have both discussed the general tensorial form of the electromagnetic boundary conditions at an arbitrary space--time interface, and given a general framework for describing these interfaces where dispersive effects are important, showing that our boundary condition (\ref{Bdd_Gen_Bdd}) encompasses all possible linear BC which conform to the conditions prescribed in section \ref{Sec:GenBddAxioms}. This gives a framework for future researchers to propose  ever increasingly interesting BCs, and can prescribe recipes for their manufacture.

We discussed the meaning of causality in the context of boundary conditions. In particular the requirement that information, such as polaritons, can only propagate in the medium, when it is temporal.

This raises some interesting questions.  The above formalism can be applied to, for instance, the problem of a `space--time corner', where the boundary has a discontinuous normal as in Figure \ref{fig:schematic}(a). This also applies to  the `space--time' prism discussed recently~\cite{li2023}.  At this point of discontinuity the conventional BCs may contradict one another, something known to be problematic in spatial BCs containing sharp corners~\cite{dobrzynski1972,pendry2017compacted}.  Extensions of the standard wedge diffraction problems~\cite{sommerfeld2004,malyuzhinets1955} (see~\cite{osipov1999,nethercote2020})  to space--time boundaries do not yet seem to have been explored.

Although we have specified the most general linear boundary conditions, we have not addressed the question as the the existence of solutions. Maxwell's equations are linear, and so are the boundary conditions presented here. Assuming that the constitutive relations for the two bulk media are also linear, then we know there is always the zero solution. However, we note that the existence of a physical, non-zero solution starting with initial conditions is not guaranteed. A simple example is a time boundary where the medium switches from being dispersive \eqref{Intr_Temp_disp_Fourier} to non-dispersive.  One can consider two options.  The first, arising from the differential equations gives, for example, $\VD^\II=\VD^I$, $\VB^\II=\VB^I$ and $\VP^\I=0$ and $\partial_t \VP^\I=0$ on the boundary. However, causality requires that the switch in the material properties can only affect the fields \emph{after} it has occurred, which is inconsistent with demanding $\VP^\I=0$ and $\partial_t \VP^\I=0$.  By contrast, one can set the BCs to be $\VD^\II=\VD^I$ and $\VB^\II=\VB^I$. In this case the boundary conditions are consistent but all the information about $\VP^\I$ is lost. 

Our formalism includes currents, such as in \eqref{eq:sheet-bc}, which are linear in the local fields. One can easily incorporate prescribed currents, by simply adding them to the right hand side of \eqref{Bdd_Gen_Bdd}. This can be used to define non trivial problems where information is given on the hypersurface boundary of 4--dimensional volume. This generalises the Dirichlet and von Neumann boundary conditions to higher order and extended dependence within the boundary. 

In principle our method can be extended to consider nonlinear BC's to any finite order wherein the kernel is integrated over each variable an appropriate number of times. For example, for a cubic response one can replace the term in \eqref{Bdd_Gen_Bdd} given by
\begin{align}
   \int_{J^\Sigma(\uvec)} d^3\uvec'
\kappa^{\F\I }_{A0}{}^{\mu\nu}(\uvec,\uvec')
\,
F^\I_{\mu\nu}(\uvec',0)
\end{align}
with
\begin{align}
\int_{J^\Sigma(\uvec)} d^3\uvec'_1
\int_{J^\Sigma(\uvec)} d^3\uvec'_2
\int_{J^\Sigma(\uvec)} d^3\uvec'_3
\kappa^{\F\I }_{A000}{}^{\mu\nu}(\uvec,\uvec'_1,\uvec'_2,\uvec'_3)
\,
F^\I_{\mu\nu}(\uvec'_1,0)
F^\I_{\mu\nu}(\uvec'_2,0)
F^\I_{\mu\nu}(\uvec'_3,0)
\end{align}
and similarly for higher derivatives.  These topics will form the focus of future research.

\subsection*{Author contributions}
JG conceptualised the ideas produced the first draft.
All authors contributed equally to writing the manuscript.

\subsection*{Funding}
JG  acknowledges support from STFC and the Cockcroft institute (ST/V001612/1).  
SARH acknowledges support from the EPSRC via the META4D Programme Grant (EP/Y015673/1).

\subsection*{Acknowledgements}
We acknowledge useful discussions with Francisco Rodriguez Fortuno.

\subsection*{Data statement}
All data generated for this article is contained within article.

\subsection*{Ethics statement}
This work did not involve human participants, animal subjects, or sensitive data, and therefore did not require ethical approval. No ethical issues were identified in the course of this research.

\bibliography{refs-edit.bib}

\begin{thebibliography}{48}%
\makeatletter
\providecommand \@ifxundefined [1]{%
 \@ifx{#1\undefined}
}%
\providecommand \@ifnum [1]{%
 \ifnum #1\expandafter \@firstoftwo
 \else \expandafter \@secondoftwo
 \fi
}%
\providecommand \@ifx [1]{%
 \ifx #1\expandafter \@firstoftwo
 \else \expandafter \@secondoftwo
 \fi
}%
\providecommand \natexlab [1]{#1}%
\providecommand \enquote  [1]{``#1''}%
\providecommand \bibnamefont  [1]{#1}%
\providecommand \bibfnamefont [1]{#1}%
\providecommand \citenamefont [1]{#1}%
\providecommand \href@noop [0]{\@secondoftwo}%
\providecommand \href [0]{\begingroup \@sanitize@url \@href}%
\providecommand \@href[1]{\@@startlink{#1}\@@href}%
\providecommand \@@href[1]{\endgroup#1\@@endlink}%
\providecommand \@sanitize@url [0]{\catcode `\\12\catcode `\$12\catcode
  `\&12\catcode `\#12\catcode `\^12\catcode `\_12\catcode `\%12\relax}%
\providecommand \@@startlink[1]{}%
\providecommand \@@endlink[0]{}%
\providecommand \url  [0]{\begingroup\@sanitize@url \@url }%
\providecommand \@url [1]{\endgroup\@href {#1}{\urlprefix }}%
\providecommand \urlprefix  [0]{URL }%
\providecommand \Eprint [0]{\href }%
\providecommand \doibase [0]{https://doi.org/}%
\providecommand \selectlanguage [0]{\@gobble}%
\providecommand \bibinfo  [0]{\@secondoftwo}%
\providecommand \bibfield  [0]{\@secondoftwo}%
\providecommand \translation [1]{[#1]}%
\providecommand \BibitemOpen [0]{}%
\providecommand \bibitemStop [0]{}%
\providecommand \bibitemNoStop [0]{.\EOS\space}%
\providecommand \EOS [0]{\spacefactor3000\relax}%
\providecommand \BibitemShut  [1]{\csname bibitem#1\endcsname}%
\let\auto@bib@innerbib\@empty
\bibitem [{\citenamefont {Mackay}\ and\ \citenamefont
  {Lakhtakia}(2019)}]{mackay2019}%
  \BibitemOpen
  \bibfield  {author} {\bibinfo {author} {\bibfnamefont {T.~G.}\ \bibnamefont
  {Mackay}}\ and\ \bibinfo {author} {\bibfnamefont {A.}~\bibnamefont
  {Lakhtakia}},\ }\href@noop {} {\emph {\bibinfo {title} {Electromagnetic
  Anisotropy and Bianisotropy: A Field Guide}}}\ (\bibinfo  {publisher} {World
  Scientific},\ \bibinfo {year} {2019})\BibitemShut {NoStop}%
\bibitem [{\citenamefont {Achouri}\ \emph {et~al.}(2015)\citenamefont
  {Achouri}, \citenamefont {Salem},\ and\ \citenamefont {Caloz}}]{achouri2015}%
  \BibitemOpen
  \bibfield  {author} {\bibinfo {author} {\bibfnamefont {K.}~\bibnamefont
  {Achouri}}, \bibinfo {author} {\bibfnamefont {M.~A.}\ \bibnamefont {Salem}},\
  and\ \bibinfo {author} {\bibfnamefont {C.}~\bibnamefont {Caloz}},\ }\bibfield
   {title} {\bibinfo {title} {General metasurface synthesis based on
  susceptibility tensors},\ }\href@noop {} {\bibfield  {journal} {\bibinfo
  {journal} {IEEE Trans. Ant. Prop.}\ }\textbf {\bibinfo {volume} {63}},\
  \bibinfo {pages} {2977} (\bibinfo {year} {2015})}\BibitemShut {NoStop}%
\bibitem [{\citenamefont {Galiffi}\ \emph {et~al.}(2022)\citenamefont
  {Galiffi}, \citenamefont {Tirole}, \citenamefont {Yin}, \citenamefont {Li},
  \citenamefont {Vezzoli}, \citenamefont {Huidobro}, \citenamefont
  {Silveirinha}, \citenamefont {Sapienza}, \citenamefont {Al{\'u}},\ and\
  \citenamefont {Pendry}}]{galiffi2022}%
  \BibitemOpen
  \bibfield  {author} {\bibinfo {author} {\bibfnamefont {E.}~\bibnamefont
  {Galiffi}}, \bibinfo {author} {\bibfnamefont {R.}~\bibnamefont {Tirole}},
  \bibinfo {author} {\bibfnamefont {S.}~\bibnamefont {Yin}}, \bibinfo {author}
  {\bibfnamefont {H.}~\bibnamefont {Li}}, \bibinfo {author} {\bibfnamefont
  {S.}~\bibnamefont {Vezzoli}}, \bibinfo {author} {\bibfnamefont {P.~A.}\
  \bibnamefont {Huidobro}}, \bibinfo {author} {\bibfnamefont {M.~G.}\
  \bibnamefont {Silveirinha}}, \bibinfo {author} {\bibfnamefont
  {R.}~\bibnamefont {Sapienza}}, \bibinfo {author} {\bibfnamefont
  {A.}~\bibnamefont {Al{\'u}}},\ and\ \bibinfo {author} {\bibfnamefont {J.~B.}\
  \bibnamefont {Pendry}},\ }\bibfield  {title} {\bibinfo {title} {Photonics of
  time-varying media},\ }\href@noop {} {\bibfield  {journal} {\bibinfo
  {journal} {Advanced Photonics}\ ,\ \bibinfo {pages} {014002}} (\bibinfo
  {year} {2022})}\BibitemShut {NoStop}%
\bibitem [{\citenamefont {Taravati}\ and\ \citenamefont
  {Eleftheriades}(2021)}]{taravati2021}%
  \BibitemOpen
  \bibfield  {author} {\bibinfo {author} {\bibfnamefont {S.}~\bibnamefont
  {Taravati}}\ and\ \bibinfo {author} {\bibfnamefont {G.~V.}\ \bibnamefont
  {Eleftheriades}},\ }\bibfield  {title} {\bibinfo {title} {Space-time
  metasurfaces: Analysis, design and applications},\ }in\ \href
  {https://doi.org/10.23919/EuCAP51087.2021.9411494} {\emph {\bibinfo
  {booktitle} {2021 15th European Conference on Antennas and Propagation
  (EuCAP)}}}\ (\bibinfo {year} {2021})\ pp.\ \bibinfo {pages}
  {1--5}\BibitemShut {NoStop}%
\bibitem [{\citenamefont {Munk}(2005)}]{munk2005}%
  \BibitemOpen
  \bibfield  {author} {\bibinfo {author} {\bibfnamefont {B.~A.}\ \bibnamefont
  {Munk}},\ }\href@noop {} {\emph {\bibinfo {title} {Frequency Selective
  Surfaces, Theory and Design}}}\ (\bibinfo  {publisher} {Wiley},\ \bibinfo
  {year} {2005})\BibitemShut {NoStop}%
\bibitem [{\citenamefont {Assouar}\ \emph {et~al.}(2018)\citenamefont
  {Assouar}, \citenamefont {Liang}, \citenamefont {Wu}, \citenamefont {Li},
  \citenamefont {Cheng},\ and\ \citenamefont {Jing}}]{assouar2018}%
  \BibitemOpen
  \bibfield  {author} {\bibinfo {author} {\bibfnamefont {B.}~\bibnamefont
  {Assouar}}, \bibinfo {author} {\bibfnamefont {B.}~\bibnamefont {Liang}},
  \bibinfo {author} {\bibfnamefont {Y.}~\bibnamefont {Wu}}, \bibinfo {author}
  {\bibfnamefont {Y.}~\bibnamefont {Li}}, \bibinfo {author} {\bibfnamefont
  {J.-C.}\ \bibnamefont {Cheng}},\ and\ \bibinfo {author} {\bibfnamefont
  {Y.}~\bibnamefont {Jing}},\ }\bibfield  {title} {\bibinfo {title} {Acoustic
  metasurfaces},\ }\href@noop {} {\bibfield  {journal} {\bibinfo  {journal}
  {Nature Reviews Materials}\ }\textbf {\bibinfo {volume} {3}},\ \bibinfo
  {pages} {460} (\bibinfo {year} {2018})}\BibitemShut {NoStop}%
\bibitem [{\citenamefont {Chen}\ \emph {et~al.}(2022)\citenamefont {Chen},
  \citenamefont {Wang}, \citenamefont {Wang}, \citenamefont {Zhou},\ and\
  \citenamefont {Yuan}}]{chen2022}%
  \BibitemOpen
  \bibfield  {author} {\bibinfo {author} {\bibfnamefont {A.-L.}\ \bibnamefont
  {Chen}}, \bibinfo {author} {\bibfnamefont {Y.-S.}\ \bibnamefont {Wang}},
  \bibinfo {author} {\bibfnamefont {Y.-F.}\ \bibnamefont {Wang}}, \bibinfo
  {author} {\bibfnamefont {H.-T.}\ \bibnamefont {Zhou}},\ and\ \bibinfo
  {author} {\bibfnamefont {S.-M.}\ \bibnamefont {Yuan}},\ }\bibfield  {title}
  {\bibinfo {title} {{Design of Acoustic/Elastic Phase Gradient Metasurfaces:
  Principles, Functional Elements, Tunability, and Coding}},\ }\href@noop {}
  {\bibfield  {journal} {\bibinfo  {journal} {Applied Mechanics Reviews}\
  }\textbf {\bibinfo {volume} {74}},\ \bibinfo {pages} {020801} (\bibinfo
  {year} {2022})}\BibitemShut {NoStop}%
\bibitem [{\citenamefont {Chen}\ \emph {et~al.}(2016)\citenamefont {Chen},
  \citenamefont {Taylor},\ and\ \citenamefont {Yu}}]{chen2016}%
  \BibitemOpen
  \bibfield  {author} {\bibinfo {author} {\bibfnamefont {H.~T.}\ \bibnamefont
  {Chen}}, \bibinfo {author} {\bibfnamefont {A.~J.}\ \bibnamefont {Taylor}},\
  and\ \bibinfo {author} {\bibfnamefont {N.}~\bibnamefont {Yu}},\ }\bibfield
  {title} {\bibinfo {title} {A review of metasurfaces: physics and
  applications},\ }\href@noop {} {\bibfield  {journal} {\bibinfo  {journal}
  {Rep. Prog. Phys.}\ }\textbf {\bibinfo {volume} {79}},\ \bibinfo {pages}
  {076401} (\bibinfo {year} {2016})}\BibitemShut {NoStop}%
\bibitem [{\citenamefont {Wang}\ \emph {et~al.}(2021)\citenamefont {Wang},
  \citenamefont {Qin},\ and\ \citenamefont {Xu}}]{wang2021}%
  \BibitemOpen
  \bibfield  {author} {\bibinfo {author} {\bibfnamefont {J.}~\bibnamefont
  {Wang}}, \bibinfo {author} {\bibfnamefont {L.}~\bibnamefont {Qin}},\ and\
  \bibinfo {author} {\bibfnamefont {W.}~\bibnamefont {Xu}},\ }\bibfield
  {title} {\bibinfo {title} {Flexible and high precision thermal metasurface},\
  }\href@noop {} {\bibfield  {journal} {\bibinfo  {journal} {Communications
  Materials}\ }\textbf {\bibinfo {volume} {2}},\ \bibinfo {pages} {89}
  (\bibinfo {year} {2021})}\BibitemShut {NoStop}%
\bibitem [{\citenamefont {Achouri}\ and\ \citenamefont
  {Caloz}(2021)}]{achouri2021}%
  \BibitemOpen
  \bibfield  {author} {\bibinfo {author} {\bibfnamefont {K.}~\bibnamefont
  {Achouri}}\ and\ \bibinfo {author} {\bibfnamefont {C.}~\bibnamefont
  {Caloz}},\ }\href@noop {} {\emph {\bibinfo {title} {Electromagnetic
  Metasurfaces: Theory and Applications}}}\ (\bibinfo  {publisher} {John Wiley
  and Sons},\ \bibinfo {year} {2021})\BibitemShut {NoStop}%
\bibitem [{\citenamefont {Tretyakov}(2003)}]{tretyakov2003}%
  \BibitemOpen
  \bibfield  {author} {\bibinfo {author} {\bibfnamefont {S.~A.}\ \bibnamefont
  {Tretyakov}},\ }\href@noop {} {\emph {\bibinfo {title} {Analytical Modeling
  in Applied Electromagnetics}}}\ (\bibinfo  {publisher} {Artech House},\
  \bibinfo {year} {2003})\BibitemShut {NoStop}%
\bibitem [{\citenamefont {Wu}\ \emph {et~al.}(2018)\citenamefont {Wu},
  \citenamefont {Coquet}, \citenamefont {Wang},\ and\ \citenamefont
  {Genevet}}]{wu2018}%
  \BibitemOpen
  \bibfield  {author} {\bibinfo {author} {\bibfnamefont {K.}~\bibnamefont
  {Wu}}, \bibinfo {author} {\bibfnamefont {P.}~\bibnamefont {Coquet}}, \bibinfo
  {author} {\bibfnamefont {Q.~J.}\ \bibnamefont {Wang}},\ and\ \bibinfo
  {author} {\bibfnamefont {P.}~\bibnamefont {Genevet}},\ }\bibfield  {title}
  {\bibinfo {title} {Modelling of free-form conformal metasurfaces},\
  }\href@noop {} {\bibfield  {journal} {\bibinfo  {journal} {Nat. Comm.}\
  }\textbf {\bibinfo {volume} {9}},\ \bibinfo {pages} {3494} (\bibinfo {year}
  {2018})}\BibitemShut {NoStop}%
\bibitem [{\citenamefont {Lebbe}\ \emph {et~al.}(2023)\citenamefont {Lebbe},
  \citenamefont {Maurel},\ and\ \citenamefont {Pham}}]{lebbe2023}%
  \BibitemOpen
  \bibfield  {author} {\bibinfo {author} {\bibfnamefont {N.}~\bibnamefont
  {Lebbe}}, \bibinfo {author} {\bibfnamefont {A.}~\bibnamefont {Maurel}},\ and\
  \bibinfo {author} {\bibfnamefont {K.}~\bibnamefont {Pham}},\ }\bibfield
  {title} {\bibinfo {title} {Homogenized transition conditions for plasmonic
  metasurfaces},\ }\href@noop {} {\bibfield  {journal} {\bibinfo  {journal}
  {Phys. Rev. B}\ }\textbf {\bibinfo {volume} {107}},\ \bibinfo {pages}
  {085124} (\bibinfo {year} {2023})}\BibitemShut {NoStop}%
\bibitem [{\citenamefont {Holloway}\ \emph {et~al.}(2009)\citenamefont
  {Holloway}, \citenamefont {Dienstfrey}, \citenamefont {Kuester},
  \citenamefont {O'Hara}, \citenamefont {Azad},\ and\ \citenamefont
  {Taylor}}]{holloway2009}%
  \BibitemOpen
  \bibfield  {author} {\bibinfo {author} {\bibfnamefont {C.~J.}\ \bibnamefont
  {Holloway}}, \bibinfo {author} {\bibfnamefont {A.}~\bibnamefont
  {Dienstfrey}}, \bibinfo {author} {\bibfnamefont {E.~F.}\ \bibnamefont
  {Kuester}}, \bibinfo {author} {\bibfnamefont {J.~F.}\ \bibnamefont {O'Hara}},
  \bibinfo {author} {\bibfnamefont {A.~K.}\ \bibnamefont {Azad}},\ and\
  \bibinfo {author} {\bibfnamefont {A.~J.}\ \bibnamefont {Taylor}},\ }\bibfield
   {title} {\bibinfo {title} {A discussion on the interpretation and
  characterization of metafilms/metasurfaces: The two-dimensional equivalent of
  metamaterials},\ }\href@noop {} {\bibfield  {journal} {\bibinfo  {journal}
  {Metamaterials}\ }\textbf {\bibinfo {volume} {3}},\ \bibinfo {pages} {100}
  (\bibinfo {year} {2009})}\BibitemShut {NoStop}%
\bibitem [{\citenamefont {Zalu{\u{s}}ki}\ \emph {et~al.}(2016)\citenamefont
  {Zalu{\u{s}}ki}, \citenamefont {Grbic},\ and\ \citenamefont
  {Hrabar}}]{zaluski2016}%
  \BibitemOpen
  \bibfield  {author} {\bibinfo {author} {\bibfnamefont {D.}~\bibnamefont
  {Zalu{\u{s}}ki}}, \bibinfo {author} {\bibfnamefont {A.}~\bibnamefont
  {Grbic}},\ and\ \bibinfo {author} {\bibfnamefont {S.}~\bibnamefont
  {Hrabar}},\ }\bibfield  {title} {\bibinfo {title} {Analytical and
  experimental characterization of metasurfaces with normal polarizability},\
  }\href@noop {} {\bibfield  {journal} {\bibinfo  {journal} {Phys. Rev. B}\
  }\textbf {\bibinfo {volume} {93}},\ \bibinfo {pages} {155156} (\bibinfo
  {year} {2016})}\BibitemShut {NoStop}%
\bibitem [{\citenamefont {Senior}\ and\ \citenamefont
  {Volakis}(1995)}]{senior1995}%
  \BibitemOpen
  \bibfield  {author} {\bibinfo {author} {\bibfnamefont {T.~B.~A.}\
  \bibnamefont {Senior}}\ and\ \bibinfo {author} {\bibfnamefont {J.~L.}\
  \bibnamefont {Volakis}},\ }\href@noop {} {\emph {\bibinfo {title}
  {Approximate boundary conditions in electromagnetics}}}\ (\bibinfo
  {publisher} {IEE Publication Series},\ \bibinfo {year} {1995})\BibitemShut
  {NoStop}%
\bibitem [{\citenamefont {Maci}\ \emph {et~al.}(2011)\citenamefont {Maci},
  \citenamefont {Minatti}, \citenamefont {Casaletti},\ and\ \citenamefont
  {Bosiljevac}}]{maci2011}%
  \BibitemOpen
  \bibfield  {author} {\bibinfo {author} {\bibfnamefont {S.}~\bibnamefont
  {Maci}}, \bibinfo {author} {\bibfnamefont {G.}~\bibnamefont {Minatti}},
  \bibinfo {author} {\bibfnamefont {M.}~\bibnamefont {Casaletti}},\ and\
  \bibinfo {author} {\bibfnamefont {M.}~\bibnamefont {Bosiljevac}},\ }\bibfield
   {title} {\bibinfo {title} {Metasurfing: Addressing waves on impenetrable
  metasurfaces},\ }\href@noop {} {\bibfield  {journal} {\bibinfo  {journal}
  {IEEE Antennas and Wireless Propagation Letters}\ }\textbf {\bibinfo {volume}
  {10}},\ \bibinfo {pages} {1499} (\bibinfo {year} {2011})}\BibitemShut
  {NoStop}%
\bibitem [{\citenamefont {Gustafson}\ and\ \citenamefont
  {Abe}(1998)}]{gustafson1998}%
  \BibitemOpen
  \bibfield  {author} {\bibinfo {author} {\bibfnamefont {K.}~\bibnamefont
  {Gustafson}}\ and\ \bibinfo {author} {\bibfnamefont {T.}~\bibnamefont
  {Abe}},\ }\bibfield  {title} {\bibinfo {title} {The third boundary
  condition—was it {R}obin’s?},\ }\href@noop {} {\bibfield  {journal}
  {\bibinfo  {journal} {Mathematical Intelligencer}\ }\textbf {\bibinfo
  {volume} {20}},\ \bibinfo {pages} {63} (\bibinfo {year} {1998})}\BibitemShut
  {NoStop}%
\bibitem [{\citenamefont {Lindell}\ and\ \citenamefont
  {Sihvola}(2017)}]{lindell2017}%
  \BibitemOpen
  \bibfield  {author} {\bibinfo {author} {\bibfnamefont {I.~V.}\ \bibnamefont
  {Lindell}}\ and\ \bibinfo {author} {\bibfnamefont {A.}~\bibnamefont
  {Sihvola}},\ }\bibfield  {title} {\bibinfo {title} {Electromagnetic wave
  reflection from boundaries defined by general linear and local conditions},\
  }\href@noop {} {\bibfield  {journal} {\bibinfo  {journal} {IEEE Transactions
  on Antennas and Propagation}\ }\textbf {\bibinfo {volume} {65}},\ \bibinfo
  {pages} {4656} (\bibinfo {year} {2017})}\BibitemShut {NoStop}%
\bibitem [{\citenamefont {Fleury}\ \emph {et~al.}(2014)\citenamefont {Fleury},
  \citenamefont {Sounas},\ and\ \citenamefont {Al{\'u'}}}]{fleury2014}%
  \BibitemOpen
  \bibfield  {author} {\bibinfo {author} {\bibfnamefont {R.}~\bibnamefont
  {Fleury}}, \bibinfo {author} {\bibfnamefont {D.~L.}\ \bibnamefont {Sounas}},\
  and\ \bibinfo {author} {\bibfnamefont {A.}~\bibnamefont {Al{\'u'}}},\
  }\bibfield  {title} {\bibinfo {title} {Negative refraction and planar
  focusing based on parity-time symmetric metasurfaces},\ }\href@noop {}
  {\bibfield  {journal} {\bibinfo  {journal} {Phys. Rev. Lett.}\ }\textbf
  {\bibinfo {volume} {113}},\ \bibinfo {pages} {023903} (\bibinfo {year}
  {2014})}\BibitemShut {NoStop}%
\bibitem [{\citenamefont {Tapar}\ and\ \citenamefont
  {Kishen}(2021)}]{tapar2021}%
  \BibitemOpen
  \bibfield  {author} {\bibinfo {author} {\bibfnamefont {J.}~\bibnamefont
  {Tapar}}\ and\ \bibinfo {author} {\bibfnamefont {N.~K.}\ \bibnamefont
  {Kishen}, \bibfnamefont {S.~andEmani}},\ }\bibfield  {title} {\bibinfo
  {title} {Dynamically tunable asymmetric transmission in pt-symmetric phase
  gradient metasurface},\ }\href@noop {} {\bibfield  {journal} {\bibinfo
  {journal} {ACS Photonics}\ }\textbf {\bibinfo {volume} {8}},\ \bibinfo
  {pages} {3315} (\bibinfo {year} {2021})}\BibitemShut {NoStop}%
\bibitem [{\citenamefont {Y.}\ \emph {et~al.}(2022)\citenamefont {Y.},
  \citenamefont {Liang}, \citenamefont {Li}, \citenamefont {Tsai},\ and\
  \citenamefont {Zhang}}]{fan2022}%
  \BibitemOpen
  \bibfield  {author} {\bibinfo {author} {\bibfnamefont {F.}~\bibnamefont
  {Y.}}, \bibinfo {author} {\bibfnamefont {H.}~\bibnamefont {Liang}}, \bibinfo
  {author} {\bibfnamefont {J.}~\bibnamefont {Li}}, \bibinfo {author}
  {\bibfnamefont {D.~P.}\ \bibnamefont {Tsai}},\ and\ \bibinfo {author}
  {\bibfnamefont {S.}~\bibnamefont {Zhang}},\ }\bibfield  {title} {\bibinfo
  {title} {Emerging trend in unconventional metasurfaces: From nonlinear,
  non-hermitian to nonclassical metasurfaces},\ }\href@noop {} {\bibfield
  {journal} {\bibinfo  {journal} {ACS Photonics}\ }\textbf {\bibinfo {volume}
  {9}},\ \bibinfo {pages} {2872} (\bibinfo {year} {2022})}\BibitemShut
  {NoStop}%
\bibitem [{\citenamefont {Zhang}\ \emph {et~al.}(2018)\citenamefont {Zhang},
  \citenamefont {Chen}, \citenamefont {Liu}, \citenamefont {Zhang},
  \citenamefont {Zhao}, \citenamefont {Dai}, \citenamefont {Bai}, \citenamefont
  {Wan}, \citenamefont {Cheng}, \citenamefont {Castaldi}, \citenamefont
  {Galdi},\ and\ \citenamefont {Cui}}]{zhang2018}%
  \BibitemOpen
  \bibfield  {author} {\bibinfo {author} {\bibfnamefont {L.}~\bibnamefont
  {Zhang}}, \bibinfo {author} {\bibfnamefont {X.~Q.}\ \bibnamefont {Chen}},
  \bibinfo {author} {\bibfnamefont {S.}~\bibnamefont {Liu}}, \bibinfo {author}
  {\bibfnamefont {Q.}~\bibnamefont {Zhang}}, \bibinfo {author} {\bibfnamefont
  {J.}~\bibnamefont {Zhao}}, \bibinfo {author} {\bibfnamefont {J.~Y.}\
  \bibnamefont {Dai}}, \bibinfo {author} {\bibfnamefont {G.~D.}\ \bibnamefont
  {Bai}}, \bibinfo {author} {\bibfnamefont {X.}~\bibnamefont {Wan}}, \bibinfo
  {author} {\bibfnamefont {Q.}~\bibnamefont {Cheng}}, \bibinfo {author}
  {\bibfnamefont {G.}~\bibnamefont {Castaldi}}, \bibinfo {author}
  {\bibfnamefont {V.}~\bibnamefont {Galdi}},\ and\ \bibinfo {author}
  {\bibfnamefont {T.~J.}\ \bibnamefont {Cui}},\ }\bibfield  {title} {\bibinfo
  {title} {Space-time-coding digital metasurfaces},\ }\href@noop {} {\bibfield
  {journal} {\bibinfo  {journal} {Nature Comm.}\ }\textbf {\bibinfo {volume}
  {9}},\ \bibinfo {pages} {4334} (\bibinfo {year} {2018})}\BibitemShut
  {NoStop}%
\bibitem [{\citenamefont {Li}\ \emph {et~al.}(2020)\citenamefont {Li},
  \citenamefont {Li}, \citenamefont {Long}, \citenamefont {Forati},
  \citenamefont {Du},\ and\ \citenamefont {Sievenpiper}}]{li2020}%
  \BibitemOpen
  \bibfield  {author} {\bibinfo {author} {\bibfnamefont {A.}~\bibnamefont
  {Li}}, \bibinfo {author} {\bibfnamefont {Y.}~\bibnamefont {Li}}, \bibinfo
  {author} {\bibfnamefont {J.}~\bibnamefont {Long}}, \bibinfo {author}
  {\bibfnamefont {E.}~\bibnamefont {Forati}}, \bibinfo {author} {\bibfnamefont
  {Z.}~\bibnamefont {Du}},\ and\ \bibinfo {author} {\bibfnamefont
  {D.}~\bibnamefont {Sievenpiper}},\ }\bibfield  {title} {\bibinfo {title}
  {Time-modulated nonreciprocal metasurface absorber for surface waves},\
  }\href@noop {} {\bibfield  {journal} {\bibinfo  {journal} {Optics Letters}\
  }\textbf {\bibinfo {volume} {45}},\ \bibinfo {pages} {1212} (\bibinfo {year}
  {2020})}\BibitemShut {NoStop}%
\bibitem [{\citenamefont {Wang}\ \emph {et~al.}(2020)\citenamefont {Wang},
  \citenamefont {D{\'i}az-Rubio}, \citenamefont {Li}, \citenamefont
  {Tretyakov},\ and\ \citenamefont {Al{\'u}}}]{wang2020}%
  \BibitemOpen
  \bibfield  {author} {\bibinfo {author} {\bibfnamefont {X.}~\bibnamefont
  {Wang}}, \bibinfo {author} {\bibfnamefont {A.}~\bibnamefont
  {D{\'i}az-Rubio}}, \bibinfo {author} {\bibfnamefont {H.}~\bibnamefont {Li}},
  \bibinfo {author} {\bibfnamefont {S.~A.}\ \bibnamefont {Tretyakov}},\ and\
  \bibinfo {author} {\bibfnamefont {A.}~\bibnamefont {Al{\'u}}},\ }\bibfield
  {title} {\bibinfo {title} {Theory and design of multifunctional space-time
  metasurfaces},\ }\href@noop {} {\bibfield  {journal} {\bibinfo  {journal}
  {Phys. Rev. Appl.}\ }\textbf {\bibinfo {volume} {13}},\ \bibinfo {pages}
  {044040} (\bibinfo {year} {2020})}\BibitemShut {NoStop}%
\bibitem [{\citenamefont {Oue}\ \emph {et~al.}(2023)\citenamefont {Oue},
  \citenamefont {Ding},\ and\ \citenamefont {Pendry}}]{oue2023}%
  \BibitemOpen
  \bibfield  {author} {\bibinfo {author} {\bibfnamefont {D.}~\bibnamefont
  {Oue}}, \bibinfo {author} {\bibfnamefont {K.}~\bibnamefont {Ding}},\ and\
  \bibinfo {author} {\bibfnamefont {J.~B.}\ \bibnamefont {Pendry}},\ }\bibfield
   {title} {\bibinfo {title} {Noncontact frictional force between surfaces by
  peristaltic permittivity modulation},\ }\href@noop {} {\bibfield  {journal}
  {\bibinfo  {journal} {Physical Review A}\ }\textbf {\bibinfo {volume}
  {107}},\ \bibinfo {pages} {063501} (\bibinfo {year} {2023})}\BibitemShut
  {NoStop}%
\bibitem [{\citenamefont {Harwood}\ \emph {et~al.}(2024)\citenamefont
  {Harwood}, \citenamefont {Vezzoli}, \citenamefont {Raziman}, \citenamefont
  {Hooper}, \citenamefont {Tirole}, \citenamefont {Wu}, \citenamefont {Maier},
  \citenamefont {Pendry}, \citenamefont {Horsley},\ and\ \citenamefont
  {Sapienza}}]{harwood2024}%
  \BibitemOpen
  \bibfield  {author} {\bibinfo {author} {\bibfnamefont {A.~C.}\ \bibnamefont
  {Harwood}}, \bibinfo {author} {\bibfnamefont {S.}~\bibnamefont {Vezzoli}},
  \bibinfo {author} {\bibfnamefont {T.~V.}\ \bibnamefont {Raziman}}, \bibinfo
  {author} {\bibfnamefont {C.}~\bibnamefont {Hooper}}, \bibinfo {author}
  {\bibfnamefont {R.}~\bibnamefont {Tirole}}, \bibinfo {author} {\bibfnamefont
  {F.}~\bibnamefont {Wu}}, \bibinfo {author} {\bibfnamefont {S.}~\bibnamefont
  {Maier}}, \bibinfo {author} {\bibfnamefont {J.~B.}\ \bibnamefont {Pendry}},
  \bibinfo {author} {\bibfnamefont {S.~A.~R.}\ \bibnamefont {Horsley}},\ and\
  \bibinfo {author} {\bibfnamefont {R.}~\bibnamefont {Sapienza}},\ }\bibfield
  {title} {\bibinfo {title} {Super-luminal synthetic motion with a space-time
  optical metasurface},\ }\href@noop {} {\bibfield  {journal} {\bibinfo
  {journal} {arXiv:2407.10809}\ } (\bibinfo {year} {2024})}\BibitemShut
  {NoStop}%
\bibitem [{\citenamefont {Caloz}\ and\ \citenamefont
  {Deck-L{\'e}ger}(2020)}]{caloz2019a}%
  \BibitemOpen
  \bibfield  {author} {\bibinfo {author} {\bibfnamefont {C.}~\bibnamefont
  {Caloz}}\ and\ \bibinfo {author} {\bibfnamefont {Z.-L.}\ \bibnamefont
  {Deck-L{\'e}ger}},\ }\bibfield  {title} {\bibinfo {title} {Spacetime
  metamaterials, part i: General concepts},\ }\href
  {https://doi.org/10.1109/TAP.2019.2944216} {\bibfield  {journal} {\bibinfo
  {journal} {IEEE Transactions on Antennas and Propagation}\ }\textbf {\bibinfo
  {volume} {68}},\ \bibinfo {pages} {1569} (\bibinfo {year}
  {2020})}\BibitemShut {NoStop}%
\bibitem [{\citenamefont {Caloz}\ and\ \citenamefont
  {Deck-L{\'e}ger}(2019)}]{caloz2019b}%
  \BibitemOpen
  \bibfield  {author} {\bibinfo {author} {\bibfnamefont {C.}~\bibnamefont
  {Caloz}}\ and\ \bibinfo {author} {\bibfnamefont {Z.-L.}\ \bibnamefont
  {Deck-L{\'e}ger}},\ }\bibfield  {title} {\bibinfo {title} {Spacetime
  metamaterials, part ii: Theory and applications},\ }\href
  {https://doi.org/10.1109/TAP.2019.2944216} {\bibfield  {journal} {\bibinfo
  {journal} {IEEE Transactions on Antennas and Propagation}\ }\textbf {\bibinfo
  {volume} {68}},\ \bibinfo {pages} {1583} (\bibinfo {year}
  {2019})}\BibitemShut {NoStop}%
\bibitem [{\citenamefont {Li}\ \emph {et~al.}(2023{\natexlab{a}})\citenamefont
  {Li}, \citenamefont {Ma}, \citenamefont {Bahrami}, \citenamefont
  {Deck-L{\'e'}ger},\ and\ \citenamefont {Caloz}}]{caloz2023}%
  \BibitemOpen
  \bibfield  {author} {\bibinfo {author} {\bibfnamefont {Z.}~\bibnamefont
  {Li}}, \bibinfo {author} {\bibfnamefont {X.}~\bibnamefont {Ma}}, \bibinfo
  {author} {\bibfnamefont {A.}~\bibnamefont {Bahrami}}, \bibinfo {author}
  {\bibfnamefont {Z.-L.}\ \bibnamefont {Deck-L{\'e'}ger}},\ and\ \bibinfo
  {author} {\bibfnamefont {C.}~\bibnamefont {Caloz}},\ }\bibfield  {title}
  {\bibinfo {title} {Space-time fresnel prism},\ }\href@noop {} {\bibfield
  {journal} {\bibinfo  {journal} {Phys. Rev. Applied}\ }\textbf {\bibinfo
  {volume} {20}},\ \bibinfo {pages} {054029} (\bibinfo {year}
  {2023}{\natexlab{a}})}\BibitemShut {NoStop}%
\bibitem [{\citenamefont {Milton}\ and\ \citenamefont
  {Mattei}(2017)}]{milton2017}%
  \BibitemOpen
  \bibfield  {author} {\bibinfo {author} {\bibfnamefont {G.~W.}\ \bibnamefont
  {Milton}}\ and\ \bibinfo {author} {\bibfnamefont {O.}~\bibnamefont
  {Mattei}},\ }\bibfield  {title} {\bibinfo {title} {Field patterns: a new
  mathematical object},\ }\href@noop {} {\bibfield  {journal} {\bibinfo
  {journal} {Proc. Roy. Soc. A}\ }\textbf {\bibinfo {volume} {473}},\ \bibinfo
  {pages} {20160816} (\bibinfo {year} {2017})}\BibitemShut {NoStop}%
\bibitem [{\citenamefont {Mostafa}\ \emph {et~al.}(2024)\citenamefont
  {Mostafa}, \citenamefont {Mirmoosa}, \citenamefont {Sidorenko}, \citenamefont
  {Asadchy},\ and\ \citenamefont {Tretyakov}}]{mostafa2024temporal}%
  \BibitemOpen
  \bibfield  {author} {\bibinfo {author} {\bibfnamefont {M.}~\bibnamefont
  {Mostafa}}, \bibinfo {author} {\bibfnamefont {M.}~\bibnamefont {Mirmoosa}},
  \bibinfo {author} {\bibfnamefont {M.}~\bibnamefont {Sidorenko}}, \bibinfo
  {author} {\bibfnamefont {V.}~\bibnamefont {Asadchy}},\ and\ \bibinfo {author}
  {\bibfnamefont {S.}~\bibnamefont {Tretyakov}},\ }\bibfield  {title} {\bibinfo
  {title} {Temporal interfaces in complex electromagnetic materials: an
  overview},\ }\href@noop {} {\bibfield  {journal} {\bibinfo  {journal}
  {Optical Materials Express}\ }\textbf {\bibinfo {volume} {14}},\ \bibinfo
  {pages} {1103} (\bibinfo {year} {2024})}\BibitemShut {NoStop}%
\bibitem [{\citenamefont {Li}\ \emph {et~al.}(2022)\citenamefont {Li},
  \citenamefont {Yin},\ and\ \citenamefont {Al{\`u}}}]{li2022nonreciprocity}%
  \BibitemOpen
  \bibfield  {author} {\bibinfo {author} {\bibfnamefont {H.}~\bibnamefont
  {Li}}, \bibinfo {author} {\bibfnamefont {S.}~\bibnamefont {Yin}},\ and\
  \bibinfo {author} {\bibfnamefont {A.}~\bibnamefont {Al{\`u}}},\ }\bibfield
  {title} {\bibinfo {title} {Nonreciprocity and faraday rotation at time
  interfaces},\ }\href@noop {} {\bibfield  {journal} {\bibinfo  {journal}
  {Physical Review Letters}\ }\textbf {\bibinfo {volume} {128}},\ \bibinfo
  {pages} {173901} (\bibinfo {year} {2022})}\BibitemShut {NoStop}%
\bibitem [{\citenamefont {Caloz}\ \emph {et~al.}(2022)\citenamefont {Caloz},
  \citenamefont {Deck-L{\'e}ger}, \citenamefont {Bahrami}, \citenamefont
  {Vicente},\ and\ \citenamefont {Li}}]{caloz2022generalized}%
  \BibitemOpen
  \bibfield  {author} {\bibinfo {author} {\bibfnamefont {C.}~\bibnamefont
  {Caloz}}, \bibinfo {author} {\bibfnamefont {Z.-L.}\ \bibnamefont
  {Deck-L{\'e}ger}}, \bibinfo {author} {\bibfnamefont {A.}~\bibnamefont
  {Bahrami}}, \bibinfo {author} {\bibfnamefont {O.~C.}\ \bibnamefont
  {Vicente}},\ and\ \bibinfo {author} {\bibfnamefont {Z.}~\bibnamefont {Li}},\
  }\bibfield  {title} {\bibinfo {title} {Generalized space-time engineered
  modulation (gstem) metamaterials},\ }\href@noop {} {\bibfield  {journal}
  {\bibinfo  {journal} {arXiv preprint arXiv:2207.06539}\ } (\bibinfo {year}
  {2022})}\BibitemShut {NoStop}%
\bibitem [{\citenamefont {Jackson}(1962)}]{jackson1962}%
  \BibitemOpen
  \bibfield  {author} {\bibinfo {author} {\bibfnamefont {J.~D.}\ \bibnamefont
  {Jackson}},\ }\href@noop {} {\emph {\bibinfo {title} {Classical
  Electrodynamics}}}\ (\bibinfo  {publisher} {John Wiley and Sons},\ \bibinfo
  {year} {1962})\BibitemShut {NoStop}%
\bibitem [{\citenamefont {Pekar}(1957)}]{pekar1957}%
  \BibitemOpen
  \bibfield  {author} {\bibinfo {author} {\bibfnamefont {S.~I.}\ \bibnamefont
  {Pekar}},\ }\bibfield  {title} {\bibinfo {title} {Dispersion of light in the
  exciton absorption region of crystals},\ }\href@noop {} {\bibfield  {journal}
  {\bibinfo  {journal} {Zh. Eksp. Teor. Fiz.}\ }\textbf {\bibinfo {volume}
  {33}},\ \bibinfo {pages} {1022} (\bibinfo {year} {1957})}\BibitemShut
  {NoStop}%
\bibitem [{\citenamefont {Gratus}\ \emph {et~al.}(2021)\citenamefont {Gratus},
  \citenamefont {Seviour}, \citenamefont {Kinsler},\ and\ \citenamefont
  {Jaroszynski}}]{gratus2021temporal}%
  \BibitemOpen
  \bibfield  {author} {\bibinfo {author} {\bibfnamefont {J.}~\bibnamefont
  {Gratus}}, \bibinfo {author} {\bibfnamefont {R.}~\bibnamefont {Seviour}},
  \bibinfo {author} {\bibfnamefont {P.}~\bibnamefont {Kinsler}},\ and\ \bibinfo
  {author} {\bibfnamefont {D.~A.}\ \bibnamefont {Jaroszynski}},\ }\bibfield
  {title} {\bibinfo {title} {Temporal boundaries in electromagnetic
  materials},\ }\href@noop {} {\bibfield  {journal} {\bibinfo  {journal} {New
  Journal of Physics}\ }\textbf {\bibinfo {volume} {23}},\ \bibinfo {pages}
  {083032} (\bibinfo {year} {2021})}\BibitemShut {NoStop}%
\bibitem [{\citenamefont {S.}\ and\ \citenamefont {Ryzhik}(2015)}]{gradshteyn}%
  \BibitemOpen
  \bibfield  {author} {\bibinfo {author} {\bibfnamefont {G.~I.}\ \bibnamefont
  {S.}}\ and\ \bibinfo {author} {\bibfnamefont {I.~M.}\ \bibnamefont
  {Ryzhik}},\ }\href@noop {} {\emph {\bibinfo {title} {Table of Integrals,
  Series, and Products}}}\ (\bibinfo  {publisher} {Academic Press},\ \bibinfo
  {year} {2015})\BibitemShut {NoStop}%
\bibitem [{\citenamefont {Landau}\ and\ \citenamefont
  {Lifshitz}(2004)}]{volume2}%
  \BibitemOpen
  \bibfield  {author} {\bibinfo {author} {\bibfnamefont {L.~D.}\ \bibnamefont
  {Landau}}\ and\ \bibinfo {author} {\bibfnamefont {E.~M.}\ \bibnamefont
  {Lifshitz}},\ }\href@noop {} {\emph {\bibinfo {title} {The Classical Theory
  of Fields}}}\ (\bibinfo  {publisher} {Butterworth-Heinemann},\ \bibinfo
  {year} {2004})\BibitemShut {NoStop}%
\bibitem [{\citenamefont {Bahrami}\ \emph {et~al.}(2024)\citenamefont
  {Bahrami}, \citenamefont {De~Kinder}, \citenamefont {Li},\ and\ \citenamefont
  {Caloz}}]{bahrami2024}%
  \BibitemOpen
  \bibfield  {author} {\bibinfo {author} {\bibfnamefont {A.}~\bibnamefont
  {Bahrami}}, \bibinfo {author} {\bibfnamefont {K.}~\bibnamefont {De~Kinder}},
  \bibinfo {author} {\bibfnamefont {Z.}~\bibnamefont {Li}},\ and\ \bibinfo
  {author} {\bibfnamefont {C.}~\bibnamefont {Caloz}},\ }\bibfield  {title}
  {\bibinfo {title} {Space-time wedges},\ }\href@noop {} {\bibfield  {journal}
  {\bibinfo  {journal} {arxiv:410.23291v1}\ } (\bibinfo {year}
  {2024})}\BibitemShut {NoStop}%
\bibitem [{\citenamefont {Khrabustovskyi}\ \emph {et~al.}(2017)\citenamefont
  {Khrabustovskyi}, \citenamefont {Mnasri}, \citenamefont {Plum}, \citenamefont
  {Stohrer},\ and\ \citenamefont {Rockstuhl}}]{khrabustovskyi2017interface}%
  \BibitemOpen
  \bibfield  {author} {\bibinfo {author} {\bibfnamefont {A.}~\bibnamefont
  {Khrabustovskyi}}, \bibinfo {author} {\bibfnamefont {K.}~\bibnamefont
  {Mnasri}}, \bibinfo {author} {\bibfnamefont {M.}~\bibnamefont {Plum}},
  \bibinfo {author} {\bibfnamefont {C.}~\bibnamefont {Stohrer}},\ and\ \bibinfo
  {author} {\bibfnamefont {C.}~\bibnamefont {Rockstuhl}},\ }\bibfield  {title}
  {\bibinfo {title} {Interface conditions for a metamaterial with strong
  spatial dispersion},\ }\href@noop {} {\bibfield  {journal} {\bibinfo
  {journal} {arXiv preprint arXiv:1710.03676}\ } (\bibinfo {year}
  {2017})}\BibitemShut {NoStop}%
\bibitem [{\citenamefont {Li}\ \emph {et~al.}(2023{\natexlab{b}})\citenamefont
  {Li}, \citenamefont {Ma}, \citenamefont {Bahrami}, \citenamefont
  {Deck-L\'eger},\ and\ \citenamefont {Caloz}}]{li2023}%
  \BibitemOpen
  \bibfield  {author} {\bibinfo {author} {\bibfnamefont {Z.}~\bibnamefont
  {Li}}, \bibinfo {author} {\bibfnamefont {X.}~\bibnamefont {Ma}}, \bibinfo
  {author} {\bibfnamefont {A.}~\bibnamefont {Bahrami}}, \bibinfo {author}
  {\bibfnamefont {Z.-L.}\ \bibnamefont {Deck-L\'eger}},\ and\ \bibinfo {author}
  {\bibfnamefont {C.}~\bibnamefont {Caloz}},\ }\bibfield  {title} {\bibinfo
  {title} {Space-time fresnel prism},\ }\href@noop {} {\bibfield  {journal}
  {\bibinfo  {journal} {Phys. Rev. Appl.}\ }\textbf {\bibinfo {volume} {20}},\
  \bibinfo {pages} {054029} (\bibinfo {year} {2023}{\natexlab{b}})}\BibitemShut
  {NoStop}%
\bibitem [{\citenamefont {Dobrzynski}\ and\ \citenamefont
  {Maradudin}(1972)}]{dobrzynski1972}%
  \BibitemOpen
  \bibfield  {author} {\bibinfo {author} {\bibfnamefont {L.}~\bibnamefont
  {Dobrzynski}}\ and\ \bibinfo {author} {\bibfnamefont {A.~A.}\ \bibnamefont
  {Maradudin}},\ }\bibfield  {title} {\bibinfo {title} {Electrostatic edge
  modes in a dielectric wedge},\ }\href@noop {} {\bibfield  {journal} {\bibinfo
   {journal} {Physical Review B}\ }\textbf {\bibinfo {volume} {6}},\ \bibinfo
  {pages} {3810} (\bibinfo {year} {1972})}\BibitemShut {NoStop}%
\bibitem [{\citenamefont {Pendry}\ \emph {et~al.}(2017)\citenamefont {Pendry},
  \citenamefont {Huidobro}, \citenamefont {Luo},\ and\ \citenamefont
  {Galiffi}}]{pendry2017compacted}%
  \BibitemOpen
  \bibfield  {author} {\bibinfo {author} {\bibfnamefont {J.}~\bibnamefont
  {Pendry}}, \bibinfo {author} {\bibfnamefont {P.~A.}\ \bibnamefont
  {Huidobro}}, \bibinfo {author} {\bibfnamefont {Y.}~\bibnamefont {Luo}},\ and\
  \bibinfo {author} {\bibfnamefont {E.}~\bibnamefont {Galiffi}},\ }\bibfield
  {title} {\bibinfo {title} {Compacted dimensions and singular plasmonic
  surfaces},\ }\href@noop {} {\bibfield  {journal} {\bibinfo  {journal}
  {Science}\ }\textbf {\bibinfo {volume} {358}},\ \bibinfo {pages} {915}
  (\bibinfo {year} {2017})}\BibitemShut {NoStop}%
\bibitem [{\citenamefont {Sommerfeld}(2004)}]{sommerfeld2004}%
  \BibitemOpen
  \bibfield  {author} {\bibinfo {author} {\bibfnamefont {A.}~\bibnamefont
  {Sommerfeld}},\ }\bibfield  {title} {\bibinfo {title} {Mathematical theory of
  diffraction},\ }in\ \href@noop {} {\emph {\bibinfo {booktitle} {Mathematical
  Theory of Diffraction}}}\ (\bibinfo  {publisher} {Springer},\ \bibinfo {year}
  {2004})\ pp.\ \bibinfo {pages} {9--68}\BibitemShut {NoStop}%
\bibitem [{\citenamefont {Malyuzhinets}(1955)}]{malyuzhinets1955}%
  \BibitemOpen
  \bibfield  {author} {\bibinfo {author} {\bibfnamefont {D.}~\bibnamefont
  {Malyuzhinets}},\ }\bibfield  {title} {\bibinfo {title} {Radiation of sound
  from the vibrating faces of an arbitrary wedge [part ii]},\ }\href@noop {}
  {\bibfield  {journal} {\bibinfo  {journal} {Sov. Phys. Acoust}\ }\textbf
  {\bibinfo {volume} {1}},\ \bibinfo {pages} {240} (\bibinfo {year}
  {1955})}\BibitemShut {NoStop}%
\bibitem [{\citenamefont {Osipov}\ and\ \citenamefont
  {Norris}(1999)}]{osipov1999}%
  \BibitemOpen
  \bibfield  {author} {\bibinfo {author} {\bibfnamefont {A.}~\bibnamefont
  {Osipov}}\ and\ \bibinfo {author} {\bibfnamefont {A.}~\bibnamefont
  {Norris}},\ }\bibfield  {title} {\bibinfo {title} {The malyuzhinets theory
  for scattering from wedge boundaries: a review},\ }\href
  {https://doi.org/https://doi.org/10.1016/S0165-2125(98)00042-0} {\bibfield
  {journal} {\bibinfo  {journal} {Wave Motion}\ }\textbf {\bibinfo {volume}
  {29}},\ \bibinfo {pages} {313} (\bibinfo {year} {1999})}\BibitemShut
  {NoStop}%
\bibitem [{\citenamefont {Nethercote}\ \emph {et~al.}(2020)\citenamefont
  {Nethercote}, \citenamefont {Assier},\ and\ \citenamefont
  {Abrahams}}]{nethercote2020}%
  \BibitemOpen
  \bibfield  {author} {\bibinfo {author} {\bibfnamefont {M.}~\bibnamefont
  {Nethercote}}, \bibinfo {author} {\bibfnamefont {R.}~\bibnamefont {Assier}},\
  and\ \bibinfo {author} {\bibfnamefont {I.}~\bibnamefont {Abrahams}},\
  }\bibfield  {title} {\bibinfo {title} {Analytical methods for perfect wedge
  diffraction: A review},\ }\href
  {https://doi.org/https://doi.org/10.1016/j.wavemoti.2019.102479} {\bibfield
  {journal} {\bibinfo  {journal} {Wave Motion}\ }\textbf {\bibinfo {volume}
  {93}},\ \bibinfo {pages} {102479} (\bibinfo {year} {2020})}\BibitemShut
  {NoStop}%
\end{thebibliography}%

\end{document}